\documentclass[journal,10pt]{IEEEtran}

\usepackage{hyperref}
\usepackage{cite}
\usepackage{graphicx}
\usepackage{array}
\usepackage{mdwmath}
\usepackage{amssymb}
\usepackage{mdwtab}
\usepackage{stfloats}
\usepackage[tight,footnotesize]{subfigure}
\usepackage{amsmath,amsthm,amsfonts}
\usepackage{threeparttable}
\usepackage{color}
\usepackage{url}
\usepackage{algorithm}
\usepackage{algorithmic}
\usepackage{multicol}
\usepackage{epstopdf}
\usepackage{flushend}
\usepackage{epsfig}
\usepackage{xcolor}
\usepackage{textcomp}
\usepackage{bm}
\usepackage{setspace}
\usepackage{float}
\newtheorem{proposition}{Proposition}

\newtheorem{remark}{Remark}

\allowdisplaybreaks[4]

\ifCLASSINFOpdf

\else

\fi

\begin{document}
\title{Throughput Maximization for Movable Antenna Systems with Movement Delay Consideration}
\author{
	Honghao~Wang,
	Qingqing~Wu,
	Ying~Gao,
	Wen~Chen,
	Weidong~Mei,
	Guojie~Hu,
	and~Lexi~Xu
	\thanks{H. Wang is with the Department of Electronic Engineering, Shanghai Jiao Tong University, Shanghai 200240, China, and also with the Department of Electrical and Computer Engineering, University of Macau, Macau 999078, China (e-mail: mc25018@um.edu.mo). Q. Wu, Y. Gao, and W. Chen are with the Department of Electronic Engineering, Shanghai Jiao Tong University, Shanghai 200240, China (e-mail: qingqingwu@sjtu.edu.cn; yinggao@sjtu.edu.cn; wenchen@sjtu.edu.cn). W. Mei is with the National Key Laboratory of Wireless Communications, University of Electronic Science and Technology of China, Chengdu 611731, China (e-mail: wmei@uestc.edu.cn). G. Hu is with the College of Communication Engineering, Rocket Force University of Engineering, Xi’an 710025, China (e-mail: lgdxhgj@sina.com). L. Xu is with the Research Institute, China United Network Communications Corporation, Beijing 100048, China (e-mail: davidlexi@hotmail.com).}
}

\maketitle

\begin{abstract}
In this paper, we model the minimum achievable throughput within a transmission block of restricted duration and aim to maximize it in movable antenna (MA)-enabled multiuser downlink communications. Particularly, we account for the antenna moving delay caused by mechanical movement, which has not been fully considered in previous studies, and reveal the trade-off between the delay and signal-to-interference-plus-noise ratio at users. To this end, we first consider a single-user setup to analyze the necessity of antenna movement. By quantizing the virtual angles of arrival, we derive the requisite region size for antenna moving, design the initial MA position, and elucidate the relationship between quantization resolution and moving region size. Furthermore, an efficient algorithm is developed to optimize MA position via successive convex approximation, which is subsequently extended to the general multiuser setup. Numerical results demonstrate that the proposed algorithms outperform fixed-position antenna schemes and existing ones without consideration of movement delay. Additionally, our algorithms exhibit excellent adaptability and stability across various transmission block durations and moving region sizes, and are robust to different antenna moving speeds. This allows the hardware cost of MA-aided systems to be reduced by employing low rotational speed motors.
\end{abstract}

\begin{IEEEkeywords}
	Movable antenna (MA), antenna moving delay, antenna position optimization, throughput maximization.
\end{IEEEkeywords}

\section{Introduction}\label{sec_intro}
The next generation of wireless communication systems demands a dramatic increase in capacity, transitioning from single-input single-output (SISO) to multiple-input multiple-output (MIMO). By leveraging the independent or quasi-independent fading characteristics of multipath components, MIMO systems can achieve substantial gains in beamforming, spatial multiplexing, and diversity, thereby enabling the simultaneous transmission of multiple data streams within the same time-frequency resource block. This capability markedly improves spectral efficiency compared to single-antenna systems. However, the implemented antenna schemes in existing systems are predominantly fixed-position antenna (FPA) \cite{WuQQ_IRS_active_passive_trans,GaoY_activeIRS_SWIPT,PengQY_HIRS_EE}, which restricts their diversity and spatial multiplexing performance due to the underutilization of channel variation in the continuous spatial field. This results from the limitation of static and discrete antenna deployment. To harness the extra spatial degree of freedom (DoF) in wireless channels, movable antenna (MA) \cite{ZhuLP_MA_model} and fluid antenna (FA) \cite{WongKK_FAS} were conceived as innovative technologies to overcome the inherent constraints of conventional FPAs. Relevant research interests are fueled by their unparalleled flexibility and reconfigurability, which provide substantial enhancements in system performance for wireless applications. Besides, the specific implementation of antenna movement is as follows: by connecting each MA to a radio frequency (RF) chain via a flexible cable, MA positions can be adjusted in real-time within a spatial region, driven by motors or servos \cite{Basbug_MA_hardware}. Thus, unlike FPAs that remain motionless and undergo random channel variations, MAs can be strategically placed at positions with more favorable channel conditions to improve the quality-of-service (QoS). In particular, MAs offer superior signal power enhancement, interference mitigation, flexible beamforming, and spatial multiplexing capabilities compared to FPAs.

Given the aforementioned advantages, MA has garnered substantial interest from the wireless communication research community. Initial explorations concentrated on point-to-point single-user MA-enabled systems. For example, the authors of \cite{ZhuLP_MA_model} introduced an innovative mechanical MA architecture and developed a field-response model for SISO systems incorporating transmit and receive MAs. Their analysis revealed significant signal-to-noise ratio (SNR) gains provided by a single receive MA compared to its FPA counterpart under both deterministic and stochastic channels. Building upon this foundation, the research in \cite{MaWY_MIMO_MA} examined an MA-enhanced point-to-point MIMO system, demonstrating that optimizing MA positions and transmit signal covariance matrix markedly improves the channel capacity. Further investigations in \cite{ChenXT_MA_statistical} and \cite{YeYQ_FAS_statistical} extended this research to scenarios with only statistical channel state information (CSI), proposing two simplified antenna movement modes with comparable performance but reduced complexity. Diverging from the narrow-band concerns, the authors of \cite{ZhuLP_wideband_MA} investigated the application of MAs in wideband communications, underscoring their potential for broader frequency bands.

Recent advancements have significantly broadened the scope of MA-aided systems to encompass more comprehensive multiuser communication scenarios, particularly emphasizing spectrum sharing, as discussed in \cite{GaoY_multicast_MA} and \cite{WangHH_interference_MA}. These works highlighted notable improvements in interference management achieved by MAs in multicast and interference networks, respectively. Beyond spectrum sharing, extensive research has been conducted on designing antenna movement schemes. For instance, the authors in \cite{ZhengZY_twotimescale_MA} and \cite{HuGJ_twotimescale_MA} proposed two-timescale designs for MA-enabled communications. These designs optimize MA positions over a large timescale based on statistical CSI while performing beamforming based on instantaneous CSI. This approach sacrifices part of the adjustment flexibility of MAs for reduced operational complexity. Additionally, other critical investigations have focused on classical multiuser uplink and downlink communications. For uplink communications, refer to \cite{ZhuLP_uplink_MA,HuGJ_power_MA,XiaoZY_MA_BS,LiNZ_NOMA_MA}, and for downlink communications, see \cite{QinHR_downlink_MA,HuGJ_CoMP_MA,MeiWD_Graph_MA,ChengZQ_FAS_sumrate,ZhangYC_MA_hybridBF,WuYF_MA_discrete,WengCH_MA_learning}. Moreover, the authors of \cite{LyuB_MA_symbiotic} and \cite{DingJZ_MA_FD_secure_1} explored MA-enabled symbiotic radio and full-duplex communications, respectively. Expanding beyond single-technique scenarios, the studies in \cite{GaoY_WPCN_MA} and \cite{WeiX_MA_RIS,SunYN_MA_RIS} evaluated the applications of MAs in conjunction with the wireless-powered communication network (WPCN) and intelligent reflecting surface (IRS), respectively. From another perspective, some works have delved into channel estimation for MA-aided communications, as evidenced by \cite{MaWY_channel_estimation_MA,XiaoZY_MA_channel_estimation,ZhangZJ_MA_channel_estimation}.

The majority of the extant works concentrated on the scenarios in which wireless channels manifest slow variations over time. This can be ascribed to the restricted mobility of wireless terminals, e.g., machine-type communication (MTC) users in automated industries \cite{ZhuLP_survey_MA}. Notably, these studies considered the transmission block with long duration, while the distances MAs move from their initial positions to the target positions are typically several wavelengths, making the antenna moving delay negligible relative to the whole transmission block. Nevertheless, although the condition of slow channel variation is fulfilled, it does not imply that the transmission duration is as prolonged as the channel stability span. In more prevalent cases, it might be necessary to communicate with multiple users sequentially within a period during which the channel can be assumed stable. This indicates that the transmission block assigned to each user is generally of short duration. Consequently, the antenna moving delay becomes a significant factor that is non-negligible to affect the communication performance. Unfortunately, this essential scenario has not been adequately tackled yet.

Building on previous discussions, contemporary antenna position optimization algorithms in \cite{MaWY_MIMO_MA,GaoY_multicast_MA,WangHH_interference_MA,ZhengZY_twotimescale_MA,HuGJ_twotimescale_MA,ZhuLP_wideband_MA,ZhuLP_uplink_MA,HuGJ_power_MA,XiaoZY_MA_BS,LiNZ_NOMA_MA,QinHR_downlink_MA,HuGJ_CoMP_MA,MeiWD_Graph_MA} aiming to enhance signal-to-interference-plus-noise ratio (SINR), maximize achievable rate, or minimize transmit power are not suitable for scenarios involving the short-duration transmission block. The reason is that the optimized MA positions derived from these algorithms may be significantly distant from the initial antenna position, causing the antenna moving delay to consume a substantial portion of transmission block duration, thereby severely diminishing communication throughput. Intuitively, a trade-off exists between antenna moving delay and users' SINR, as MA positions that yield a sub-optimal SINR but are closer to the initial position may result in optimal throughput. Thus, an urgent need is to develop MA position optimization algorithms that adapt to different transmission block durations. In this paper, we creatively take the antenna moving delay into account for the first time and pioneered the investigation of communication throughput in MA-enabled systems within a shorter transmission block duration.

In light of the above, this paper investigates an MA-aided multiuser multiple-input single-output (MISO) downlink communication system which is composed of a BS with $N$ FPAs and $K$ single-MA users, as shown in Fig. \ref{fig1_model}. Based on this, we comprehensively reveal the trade-off between antenna moving delay and users' SINR to achieve optimal communication throughput. The main contributions of this paper are summarized as follows:
\begin{itemize}
	\item To the best of the authors' knowledge, this is the first work in the literature to investigate the impact of antenna moving delay on the performance of MA-aided communications. We formulate a minimum achievable throughput maximization problem that jointly optimizes MA positions and transmit beamforming vectors under the constraints of transmission block duration, maximum transmit power at BS, and finite moving region for each MA.
	\item To handle this challenging problem, we first consider a simplified scenario with a single user. By leveraging the inherent characteristics of multipath channels and extra spatial DoFs that MA offers, we present the judgment and condition for whether the antenna at the user needs to move when the number of channel paths is one or two. By quantizing the virtual angles of arrival (AoAs) under multipath channels, we determine the requisite region size for antenna moving, design the initial MA position, and elucidate the relationship between quantization resolution and the moving region size. Furthermore, an effective algorithm is proposed based on successive convex approximation (SCA) for iteratively optimizing the MA position. Subsequently, we extend the proposed algorithm to the general case with multiple users by introducing slack variables.
	\item Numerical results validate that our algorithms outperform both the FPA scheme and the method outlined in \cite{GaoY_multicast_MA} in terms of achievable throughput under various setups. Moreover, the proposed algorithms remain superior across varying transmission block durations, effectively addressing the limitations of existing research in this domain. Additionally, our algorithms are robust to variations in antenna moving speed and drastically enhance performance stability for different region sizes wherein MAs move, with performance fluctuations of less than $3\%$. Both of these superiorities facilitate the simplification of MA hardware design.
\end{itemize}

\begin{figure*}[t]
	\centering
	\includegraphics[width=4.8in]{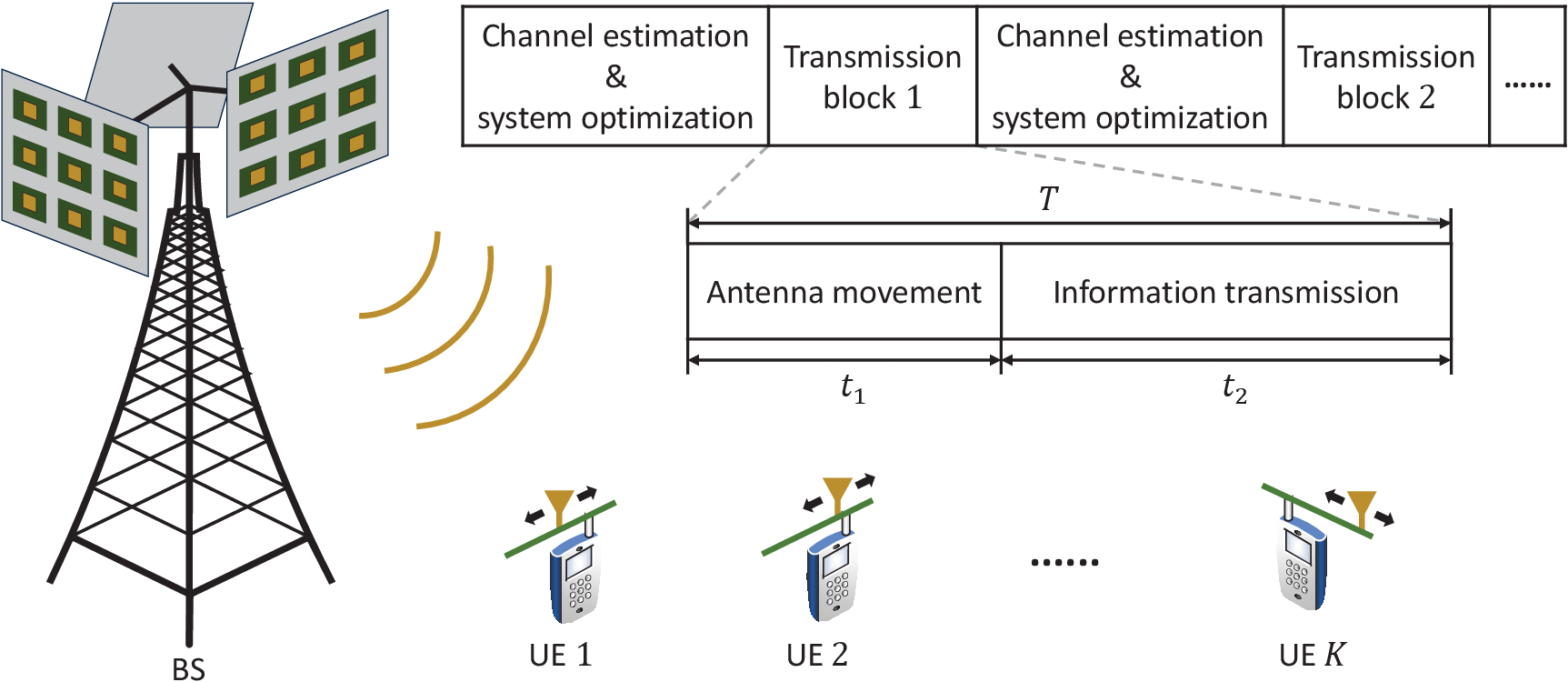}
	\vspace{-3pt}
	\caption{System model and transmission protocol of MA-enabled multiuser MISO downlink communications.}
	\label{fig1_model}
	\vspace{-6pt}
\end{figure*}

The rest of this paper is organized as follows. Section \ref{sec_model} introduces the system model and formulates the minimum achievable throughput maximization problem. Analysis of antenna movement and an efficient algorithm is proposed in Section \ref{section3} for addressing the simplified scenario with a single user, which is subsequently expanded to handle the general multiuser system in Section \ref{section4}. We evaluate the performance of our proposed algorithms via simulations in Section \ref{section5}. Finally, Section \ref{section6} concludes the paper.

\emph{Notations:} $\mathbb{R}$, $\mathbb{C}$, and $\mathbb{Z}$ denote the real space, complex space, and the set of integers, respectively. $\mathbb{C}^{M\times{N}}$ represents the space of $M\times{N}$ complex-valued matrices. For a complex-valued number $x$, $\left|x\right|$, $\operatorname{Re}\left\{x\right\}$, and $\angle{x}$ denote its modulus, real part, and phase argument, respectively. For a complex-valued vector $\mathbf{a}$, $\left|\left|\mathbf{a}\right|\right|$ and $\operatorname{diag}\left(\mathbf{a}\right)$ denote its Euclidean norm and diagonalization matrix, respectively. For a matrix $\mathbf{A}$ of arbitrary size, $\mathbf{A}\left(i,j\right)$, $\operatorname{rank}\left(\mathbf{A}\right)$, and $\operatorname{tr}\left(\mathbf{A}\right)$ represent its $\left(i,j\right)$-th element, rank, and trace, respectively. For two square matrices $\mathbf{W}_1$ and $\mathbf{W}_2$, $\mathbf{W}_1\succeq\mathbf{W}_2$ indicates that $\mathbf{W}_1-\mathbf{W}_2$ is positive semidefinite. The conjugate transpose operator is denoted by $\left(\cdot\right)^H$, while the expectation operator is represented by $\mathbb{E}\left\{\cdot\right\}$. $\mathcal{CN}\left(\mathbf{x},\mathbf{\Sigma}\right)$ represents a complex Gaussian distribution with a mean vector $\mathbf{x}$ and co-variance matrix $\mathbf{\Sigma}$.

\section{System Model and Problem Formulation}\label{sec_model}
As shown in Fig. \ref{fig1_model}, we consider a multiuser MISO downlink communication system in which a BS is equipped with $N$ rectangular FPAs and $K$ users are equipped with a single MA. Assuming that MAs connected to RF chains via flexible cables can move along linear arrays without restraint. Then, the coordinates of MAs for all users are denoted as $\mathbf{x}=\left[x_{1},\ldots,x_{K}\right]\in\mathbb{R}^{1{\times}K}$, where $x_{k}\in\mathcal{A}_k$ $\left(k\in\mathcal{K}=\left\{1,\ldots,K\right\}\right)$ represented by Cartesian coordinates represents the MA position of user $k$, and $\mathcal{A}_k$ is the given one-dimensional (1D) moving region with interval $\left[0,A_k\right]$. According to the field-response model \cite{ZhuLP_MA_model}, the channel vector $\mathbf{h}_{k}\in\mathbb{C}^{N\times1}$ from BS to user $k$ follows the structure as
\begin{equation}\label{eq_channel_1}
	\mathbf{h}_{k}=\mathbf{G}_{k}^H\mathbf{\Delta}_{k}\mathbf{f}_k\!\left(x_k\right).
\end{equation}
The terms appearing in \eqref{eq_channel_1} are defined as follows:
\begin{itemize}
	\item $\mathbf{f}_k\!\left(x_k\right)\overset{\triangle}{=}\left[e^{j\frac{2\pi}{\lambda}x_k\vartheta_{k,1}^r},\ldots,e^{j\frac{2\pi}{\lambda}x_k\vartheta_{k,L_k}^r}\right]^T$, where $L_{k}$ is the number of channel paths from the BS to user $k$, $\vartheta_{k,l}^r=\sin\theta_{k,l}^r\cos\phi_{k,l}^r,~l\in\left\{1,\ldots,L_k\right\}$ represents the virtual AoAs of the $l$-th receive path at user $k$, $\theta_{k,l}^r\in\left[0,\pi\right]$ and $\phi_{k,l}^r\in\left[0,\pi\right]$ denote the elevation and azimuth angles, respectively, and $\lambda$ is the wavelength.
	\item $\mathbf{\Delta}_{k}\overset{\triangle}{=}\operatorname{diag}\left\{\left[\tau_{k,1},\ldots,\tau_{k,L_{k}}\right]^H\right\}\in\mathbb{C}^{L_{k}\times{L_{k}}}$ indicates the path-response matrix (PRM) of user $k$, where $\tau_{k,l}$ is the complex response of the $l$-th channel path.
	\item $\mathbf{G}_{k}\overset{\triangle}{=}\left[\mathbf{g}_{k,1},\ldots,\mathbf{g}_{k,N}\right]\in\mathbb{C}^{L_{k}\times{N}}$ represents the transmit field-response matrix (FRM) at the BS, where $\mathbf{g}_{k,n}\in\mathbb{C}^{L_{k}\times1}$ is the transmit field-response vector (FRV) between the $n$-th transmit antenna and user $k$ $\left(n\in\left\{1,\ldots,N\right\}\right)$.
	\item $\mathbf{g}_{k,n}\overset{\triangle}{=}\left[e^{j\frac{2\pi}{\lambda}\mathbf{t}_{n}^T\mathbf{p}_{k,1}},\ldots,e^{j\frac{2\pi}{\lambda}\mathbf{t}_{n}^T\mathbf{p}_{k,L_k}}\right]^T$, where $\mathbf{t}_n\in\mathbb{R}^{2\times1}$ denotes the position of the $n$-th transmit antenna, $\mathbf{p}_{k,l}\overset{\triangle}{=}\left[\sin\theta_{k,l}^t\cos\phi_{k,l}^t,\cos\theta_{k,l}^t\right]^T$, $\theta_{k,l}^t\in\left[0,\pi\right]$ and $\phi_{kj,l}^t\in\left[0,\pi\right]$ are the elevation and azimuth angles of departure (AoDs) of the $l$-th transmit path from the BS to user $k$, respectively.
\end{itemize}
Based on the above, equation \eqref{eq_channel_1} can be reconstructed as
\begin{equation}\label{eq_channel_2}
	\mathbf{h}_{k}\!\left(x_k\right)=\left[\begin{matrix}
		\sum_{l=1}^{L_{k}}\tau_{k,l}^*e^{j\frac{2\pi}{\lambda}\left(x_k\vartheta_{k,l}^r-\mathbf{t}_1^T\mathbf{p}_{k,l}\right)}\\[-5pt]
		\vdots\\
		\sum_{l=1}^{L_{k}}\tau_{k,l}^*e^{j\frac{2\pi}{\lambda}\left(x_k\vartheta_{k,l}^r-\mathbf{t}_N^T\mathbf{p}_{k,l}\right)}
	\end{matrix}\right].
\end{equation}

Let $\mathbf{w}_k\in\mathbb{C}^{N\times1}$ denote the transmit beamforming vector for user $k$, the baseband complex signal received at user $k$ can be then expressed as
\begin{equation}\label{eq_rec}
	\small
	y_k=\sum_{j=1}^{K}\mathbf{h}_{k}^H\mathbf{w}_{j}s_j+z_k=\mathbf{h}_{k}^H\mathbf{w}_ks_k+\sum_{j\ne{k}}\mathbf{h}_{k}^H\mathbf{w}_{j}s_j+z_k,~\forall{k}\in\mathcal{K},
\end{equation}
where $s_k$ is the information-bearing symbol for user $k$ with normalized power, $z_k$ denotes the additive white Gaussian noise (AWGN) at user $k$, which is assumed to be circularly symmetric complex Gaussian (CSCG) distributed with zero mean and power $\sigma_k^2$, i.e., $z_k\sim\mathcal{CN}\left(0,\sigma_k^2\right)$. Then, the resulting SINR of user $k$ can be written as
\begin{equation}\label{eq_SINR}
	\gamma_k=\frac{\left|\mathbf{h}_{k}^H\mathbf{w}_k\right|^2}{\sum_{j=1,j\ne{k}}^K\left|\mathbf{h}_{k}^H\mathbf{w}_j\right|^2+\sigma_k^2},~\forall{k}\in\mathcal{K}.
\end{equation}

\begin{figure*}[hb]
	\newcounter{al1}
	\setcounter{al1}{\value{equation}}
	\setcounter{equation}{9}
	\hrulefill
	\begin{align}\label{eq_h^2}
		\left|\left|\mathbf{h}\right|\right|^2&=\mathbf{f}\!\left(x\right)^H\mathbf{\Delta}^H\mathbf{G}\mathbf{G}^H\mathbf{\Delta}\mathbf{f}\!\left(x\right)=\sum_{a=1}^{L-1}\sum_{b=a+1}^{L}2\operatorname{Re}\left\{\tau_a^*\tau_be^{j\frac{2\pi}{\lambda}\left(\vartheta_b^r-\vartheta_a^r\right)x}\sum_{n=1}^Ne^{j\frac{2\pi}{\lambda}\mathbf{t}_n^T\left(\mathbf{p}_a-\mathbf{p}_b\right)}\right\}+N\sum_{l=1}^{L}\left|\tau_l\right|^2\nonumber\\
		&=\sum_{a=1}^{L-1}\sum_{b=a+1}^{L}2\left|\tau_a^*\tau_b\sum_{n=1}^Ne^{j\frac{2\pi}{\lambda}\mathbf{t}_n^T\left(\mathbf{p}_a-\mathbf{p}_b\right)}\right|\cos\left(\frac{2\pi}{\lambda}\left(\vartheta_b^r-\vartheta_a^r\right)x+\angle\Big(\tau_a^*\tau_b\sum_{n=1}^Ne^{j\frac{2\pi}{\lambda}\mathbf{t}_n^T\left(\mathbf{p}_a-\mathbf{p}_b\right)}\Big)\right)+N\sum_{l=1}^{L}\left|\tau_l\right|^2.
	\end{align}
	\setcounter{equation}{\value{al1}}
\end{figure*}

The channel suffers from severe fluctuations when MA moves. If information transmission occurs during this period, it will result in poor QoS and thus low energy efficiency. Consequently, it is assumed that only basic signaling overhead is maintained, and no information transmission is carried out amid antenna movement. Hence, a transmission block with duration $T$ can be divided into two phases, i.e., antenna moving (AM) phase with delay $t_1$ and information transmission (IT) phase with duration $t_2$, $t_1+t_2=T$. During the AM phase, users move their antennas to positions with more favorable channel conditions to enhance the desired signal while mitigating interference. During the IT phase, users receive the data transmitted from BS aided by fixed MA. The AM delay $t_1$ is related to the distance and speed of antenna movement, which hereby can be expressed as
\begin{equation}\label{eq_t1}
	t_1=\underset{k\in\mathcal{K}}{\max}\left\{\frac{\left|x_k-x_k^0\right|}{v_k}\right\},
\end{equation}
where $x_k^0\in\mathcal{A}_k$ and $v_k$ denote the initial antenna position and antenna moving speed of user $k$, respectively. The farther the target antenna position is from the initial position, the longer the AM delay will be. This increased delay reduces the IT duration, potentially limiting communication throughput within the transmission block. Then, the achievable throughput $C_k$ (in bits/Hz) of user $k$ within the transmission block is given by\looseness=-1
\begin{equation}\label{eq_amount}
	C_k=t_2\log_2\left(1+\gamma_k\right),~\forall{k}\in\mathcal{K}.
\end{equation}

In this paper, to ensure the fairness of QoS for each user, we aim to maximize the minimum achievable throughput within the transmission block among all users by jointly optimizing AM delay and IT duration $\left\{t_1,t_2\right\}$, MA positions $\mathbf{x}$, and transmit beamforming vectors $\left\{\mathbf{w}_k\right\}_{k\in\mathcal{K}}$. According to \eqref{eq_SINR}-\eqref{eq_amount}, the optimization problem is formulated as follows:
\begin{subequations}\label{eq_P1}
	\begin{alignat}{2}
		\text{(P1)}:\quad&\underset{\mathbf{x},\left\{\mathbf{w}_k\right\}_{k\in\mathcal{K}}}{\max}~\underset{k\in\mathcal{K}}{\min}\quad&&\tilde{C}_k\label{eq_P1_a}\\
		&~\quad\mathrm{s.t.}&&\!\sum_{k=1}^{K}\left|\left|\mathbf{w}_k\right|\right|^2\le{P_{\text{m}}},\label{eq_P1_b}\\
		&&&x_k\in\mathcal{A}_k,~{\forall}k\in\mathcal{K},\label{eq_P1_c}
	\end{alignat}
\end{subequations}
where $\tilde{C}_k\overset{\triangle}{=}\left(\!T-\underset{k\in\mathcal{K}}{\max}\left\{\!\frac{\left|x_k-x_k^0\right|}{v_k}\!\right\}\!\right)\log_2\!\left(1\!+\!\gamma_k\right)$ is the reconstructed achievable throughput, $P_{\text{m}}>0$ in \eqref{eq_P1_b} represents the maximum transmit power of BS, and \eqref{eq_P1_c} ensures that MAs lie in the valid range. Despite the elimination of variables $\left\{t_1,t_2\right\}$, problem (P1) remains a non-convex optimization problem due to the intractable coupling of variables $\left\{\mathbf{x},\left\{\mathbf{w}_k\right\}_{k\in\mathcal{K}}\right\}$ in the objective function. To address this challenge, we first focus on the single-user system, laying the groundwork for tackling the more complex multiuser system. Specifically, in the next section, we present movement analyses considering different numbers of channel paths and propose an efficient algorithm to solve (P1) for the single-user setup. This approach is then extended to the multiuser setup in Section \ref{section4}.

\section{Single-User System}\label{section3}
In this section, we consider the single-user setup, i.e., $K=1$, to draw crucial insights into the joint MA position and transmit beamforming design. In this case, with the absence of inter-user interference and the elimination of differences in AM delays for each user, (P1) is thus simplified to (by omitting the user index $k$ for brevity)
\begin{subequations}\label{eq_P2}
	\begin{alignat}{2}
		\text{(P2)}:\quad&\underset{x,\mathbf{w}}{\max}\quad&&\left(T-\frac{\left|x-x^0\right|}{v}\right)\log_2\left(1+\frac{\left|\mathbf{h}^H\mathbf{w}\right|^2}{\sigma^2}\right)\label{eq_P2_a}\\
		&~\mathrm{s.t.}&&~\!\left|\left|\mathbf{w}\right|\right|^2\le{P_{\text{m}}},\label{eq_P2_b}\\
		&&&~x\in\mathcal{A}.\label{eq_P2_c}
	\end{alignat}
\end{subequations}
For any given MA positions $\mathbf{x}$, it is known that the item $\left|\mathbf{h}^H\mathbf{w}\right|^2$ can be maximized by adopting the maximum ratio transmission (MRT), i.e., $\mathbf{w}^\star=\sqrt{P_{\text{m}}}\frac{\mathbf{h}}{\left|\left|\mathbf{h}\right|\right|}$, which aligns all dimensions of $\mathbf{h}$ and is thus the optimal transmit beamforming. Then, problem (P2) is equivalently transformed into
\begin{subequations}\label{eq_P2.1}
	\begin{alignat}{2}
		\text{(P2.1)}:\quad&\underset{x}{\max}~&&\left(T\!-\!\frac{\left|x-x^0\right|}{v}\right)\!\log_2\!\left(1\!+\!\frac{P_{\text{m}}\left|\left|\mathbf{h}\right|\right|^2}{\sigma^2}\right)\label{eq_P2.1_a}\\
		&~\mathrm{s.t.}&&~\eqref{eq_P2_c}.\nonumber
	\end{alignat}
\end{subequations}
By substituting \eqref{eq_channel_1} into the item $\left|\left|\mathbf{h}\right|\right|^2$ in \eqref{eq_P2.1_a}, equation \eqref{eq_h^2} is derived as shown at the bottom of the page and problem (P2.1) can be further rewritten as
\setcounter{equation}{10}
\begin{subequations}\label{eq_P2.2}
	\begin{alignat}{2}
		\text{\rm(P2.2)}:\quad&\underset{x}{\max}\quad&&\tilde{C}\label{eq_P2.2_a}\\
		&~\mathrm{s.t.}&&\eqref{eq_P2_c},\nonumber
	\end{alignat}
\end{subequations}
where we define
{\footnotesize
\begin{equation}
	\tilde{C}\overset{\triangle}{=}\tilde{T}\log_2\!\left(\!1\!+\!P_{\text{m}}\frac{\sum_{a=1}^{L\!-\!1}\!\sum_{b=a\!+\!1}^{L}\!2\left|\mathcal{F}_{ab}\right|\cos\!\left(\!\frac{2\pi}{\lambda}\tilde{\theta}_{ba}^rx\!+\!\angle\mathcal{F}_{ab}\!\right)\!+\hspace{-1pt}\mathcal{G}}{\sigma^2}\right)\!,\nonumber
\end{equation}}
\vspace{-12pt}
{\small
\begin{align}
	&\hspace{-55pt}\mathcal{F}_{ab}\overset{\triangle}{=}\tau_a^*\tau_b\sum_{n=1}^Ne^{j\frac{2\pi}{\lambda}\mathbf{t}_n^T\left(\mathbf{p}_a-\mathbf{p}_b\right)},~a,b\in\left\{1,\ldots,L\right\},\nonumber\\
	&\hspace{-54pt}\tilde{\theta}_{ab}^r\overset{\triangle}{=}\vartheta_a^r-\vartheta_b^r,~a,b\in\left\{1,\ldots,L\right\},\nonumber\\
	&\hspace{-54pt}\tilde{T}\overset{\triangle}{=}T-\frac{\left|x-x^0\right|}{v},~\mathcal{G}\overset{\triangle}{=}N\sum_{l=1}^{L}\left|\tau_l\right|^2.\nonumber
\end{align}}

\vspace{-15pt}
\subsection{When to Move?}\label{subsec_analysis}
In the following, we give the judgment and condition of whether the antenna at the user needs to move when the number of channel paths is one or two. In the case of multiple channel paths, we derive the requisite region size for MA moving by assuming quantized virtual AoAs. This particular region may not align with the previous $\mathcal{A}$ with interval $\left[0,A\right]$, however, (P2.1) can be optimally solved within this region, where $A$ denotes the length of the initial moving region. Additionally, insightful analyses on the initial antenna position $x^0$ design and the relationship between quantization resolution and moving region size are provided. Furthermore, an efficient algorithm is proposed to solve the antenna position optimization problem (P2.2).

\begin{figure*}[!b]
	\vspace{-5pt}
	\hrulefill
	\newcounter{al2}
	\setcounter{al2}{\value{equation}}
	\setcounter{equation}{16}
	{\footnotesize
		\begin{align}
			&\frac{\operatorname{d}y_1\!\left(x\right)}{\operatorname{d}x}\!=\!-\frac{1}{v}\log_2\!\left(\!1\!+\!\frac{P_{\text{m}}\mathcal{G}}{\sigma^2}\!+\!\frac{2P_{\text{m}}}{\sigma^2}\!\left|\mathcal{F}_{12}\right|\!\cos\!\left(\!\frac{2\pi}{\lambda}\tilde{\theta}_{21}^rx\!+\!\angle\mathcal{F}_{12}\!\right)\!\right)\!-\!\left(\!T\!-\!\frac{x\!-\!x^0}{v}\!\right)\!\frac{4{\pi}P_{\text{m}}\left|\mathcal{F}_{12}\right|\tilde{\theta}_{21}^r\sin\left(\frac{2\pi}{\lambda}\tilde{\theta}_{21}^rx+\angle\mathcal{F}_{12}\right)}{\sigma^2\lambda\ln\!2\left(\!1\!+\!\frac{P_{\text{m}}\mathcal{G}}{\sigma^2}\!+\!\frac{2P_{\text{m}}}{\sigma^2}\!\left|\mathcal{F}_{12}\right|\!\cos\!\left(\!\frac{2\pi}{\lambda}\tilde{\theta}_{21}^rx\!+\!\angle\mathcal{F}_{12}\!\right)\!\right)},~\text{for}~x\ge{x^0},\label{eq_deriv_ge}\\
			&\frac{\operatorname{d}y_1\!\left(x\right)}{\operatorname{d}x}\!=\!\frac{1}{v}\log_2\!\left(\!1\!+\!\frac{P_{\text{m}}\mathcal{G}}{\sigma^2}\!+\!\frac{2P_{\text{m}}}{\sigma^2}\!\left|\mathcal{F}_{12}\right|\!\cos\!\left(\!\frac{2\pi}{\lambda}\tilde{\theta}_{21}^rx\!+\!\angle\mathcal{F}_{12}\!\right)\!\right)\!-\!\left(\!T\!-\!\frac{x^0\!-\!x}{v}\!\right)\!\frac{4{\pi}P_{\text{m}}\left|\mathcal{F}_{12}\right|\tilde{\theta}_{21}^r\sin\left(\frac{2\pi}{\lambda}\tilde{\theta}_{21}^rx+\angle\mathcal{F}_{12}\right)}{\sigma^2\lambda\ln\!2\left(\!1\!+\!\frac{P_{\text{m}}\mathcal{G}}{\sigma^2}\!+\!\frac{2P_{\text{m}}}{\sigma^2}\!\left|\mathcal{F}_{12}\right|\!\cos\!\left(\!\frac{2\pi}{\lambda}\tilde{\theta}_{21}^rx\!+\!\angle\mathcal{F}_{12}\!\right)\!\right)},~\text{for}~x<{x^0}.\label{eq_deriv_le}
	\end{align}}%
	\setcounter{equation}{\value{al2}}
\end{figure*}

\subsubsection{One-Path Case}
For the light-of-sight (LoS) channel, i.e., $L=1$, the objective function of (P2.2) is degenerated into\looseness=-1
\begin{equation}
	\small
	\tilde{C}_{L=1}=\left(T-\frac{\left|x-x^0\right|}{v}\right)\log_2\left(1+\frac{P_{\text{m}}\mathcal{G}}{\sigma^2}\right).
\end{equation}
In this case, the optimal MA position is the initial antenna coordinate $x^0$, i.e., antenna movement does not bring any channel gain but consumes time and thus reduces throughput.

\subsubsection{Two-Path Case}
If two channel paths with different virtual AoAs arrive at the user, the objective function of (P2.2) denoted by $y_1\!\left(x\right)$ is re-expressed as

\vspace{-10pt}
{\small
\begin{align}\label{eq_C_2path}
	&y_1\!\left(x\right)\overset{\triangle}{=}\tilde{C}_{L=2}\nonumber\\
	&=\tilde{T}\log_2\!\left(1+\frac{P_{\text{m}}\mathcal{G}}{\sigma^2}+\frac{2P_{\text{m}}}{\sigma^2}\left|\mathcal{F}_{12}\right|\cos\!\left(\frac{2\pi}{\lambda}\tilde{\theta}_{21}^rx\!+\!\angle\mathcal{F}_{12}\right)\!\right).
\end{align}}%

\begin{proposition}
	The maximum achievable throughput is attained at the initial antenna coordinate $x^0$ if the following sufficient conditions are met:
	
	\vspace{-8pt}
	{\small
		\begin{align}
			&\left\{\begin{aligned}
				\frac{A\tilde{\theta}_{21}^r}{\lambda}\!+\!\frac{\angle\mathcal{F}_{12}}{2\pi}\!-\!\frac{1}{2}\le{d_1}\le\frac{x^0\tilde{\theta}_{21}^r}{\lambda}\!+\!\frac{\angle\mathcal{F}_{12}}{2\pi},\vartheta_1^r<\vartheta_2^r,\\
				\frac{x^0\tilde{\theta}_{21}^r}{\lambda}\!+\!\frac{\angle\mathcal{F}_{12}}{2\pi}\le{d_1}\le\frac{A\tilde{\theta}_{21}^r}{\lambda}\!+\!\frac{\angle\mathcal{F}_{12}}{2\pi}\!+\!\frac{1}{2},\vartheta_1^r>\vartheta_2^r,
			\end{aligned}\right.\exists\hspace{1pt}d_1\in\mathbb{Z},\label{eq_yx_dec}\\[-3pt]
			&\left\{\begin{aligned}
				\frac{x^0\tilde{\theta}_{21}^r}{\lambda}+\frac{\angle\mathcal{F}_{12}}{2\pi}\le{d_2}\le\frac{\angle\mathcal{F}_{12}}{2\pi}+\frac{1}{2},~\vartheta_1^r<\vartheta_2^r,\\
				\frac{\angle\mathcal{F}_{12}}{2\pi}-\frac{1}{2}\le{d_2}\le\frac{x^0\tilde{\theta}_{21}^r}{\lambda}+\frac{\angle\mathcal{F}_{12}}{2\pi},~\vartheta_1^r>\vartheta_2^r,
			\end{aligned}\right.\exists\hspace{1pt}d_2\in\mathbb{Z},\label{eq_yx_inc}
	\end{align}}%
	indicating that antenna position optimization is unnecessary. Conversely, optimizing the MA position should be considered if the following inequality holds:
	\begin{equation}\label{eq_2path}
		A\left|\tilde{\theta}_{21}^r\right|>\lambda.
	\end{equation}
\end{proposition}

\begin{proof}
	To analyze the variation of $y_1\!\left(x\right)$ in the receive region, the gradient of which with respect to (w.r.t.) MA position $x$ is derived in \eqref{eq_deriv_ge} and \eqref{eq_deriv_le} as shown at the bottom of the page. For the case of $x\ge{x^0}$, the sufficient condition that equation \eqref{eq_deriv_ge} is always less than $0$ is given by
	\setcounter{equation}{18}
	\begin{equation}
		\small
		\tilde{\theta}_{21}^r\sin\left(\frac{2\pi}{\lambda}\tilde{\theta}_{21}^rx+\angle\mathcal{F}_{12}\right)\ge0,
		\vspace{-3pt}
	\end{equation}
	i.e.,
	{\small
	\begin{align}\label{eq_deri_1}
		\left\{\begin{aligned}
			&\!\frac{\lambda}{\tilde{\theta}_{21}^r}\!\!\left(\!d_1\!-\!\frac{\angle\mathcal{F}_{12}}{2\pi}\!\right)\!\le{x}\le\!\frac{\lambda}{\tilde{\theta}_{21}^r}\!\!\left(\!d_1\!+\!\frac{1}{2}\!-\!\frac{\angle\mathcal{F}_{12}}{2\pi}\!\right)\!,\vartheta_1^r\!<\!\vartheta_2^r,\exists\hspace{1pt}d_1\!\in\!\mathbb{Z},\\
			&\!\frac{\lambda}{\tilde{\theta}_{21}^r}\!\!\left(\!d_1\!-\!\frac{\angle\mathcal{F}_{12}}{2\pi}\!\right)\!\le{x}\le\!\frac{\lambda}{\tilde{\theta}_{21}^r}\!\!\left(\!d_1\!-\!\frac{1}{2}\!-\!\frac{\angle\mathcal{F}_{12}}{2\pi}\!\right)\!,\vartheta_1^r\!>\!\vartheta_2^r,\exists\hspace{1pt}d_1\!\in\!\mathbb{Z}.
		\end{aligned}\right.
	\end{align}}%
	If there exists an integer $d_1$ such that the following inequalities hold, the achievable throughput $y_1\!\left(x\right)$ decreases as $x$ increases within the receive region $\left[x^0,A\right]$:
	
	\vspace{-10pt}
	{\small
	\begin{align}
		\left\{\begin{aligned}
			&\frac{A\tilde{\theta}_{21}^r}{\lambda}+\frac{\angle\mathcal{F}_{12}}{2\pi}-\frac{1}{2}\le{d_1}\le\frac{x^0\tilde{\theta}_{21}^r}{\lambda}+\frac{\angle\mathcal{F}_{12}}{2\pi},~\vartheta_1^r<\vartheta_2^r,\\
			&\frac{x^0\tilde{\theta}_{21}^r}{\lambda}+\frac{\angle\mathcal{F}_{12}}{2\pi}\le{d_1}\le\frac{A\tilde{\theta}_{21}^r}{\lambda}+\frac{\angle\mathcal{F}_{12}}{2\pi}+\frac{1}{2},~\vartheta_1^r>\vartheta_2^r,
		\end{aligned}\right.
	\end{align}}%
	which indicates that any antenna movement will lead to a reduction in system throughput. Moreover, $x^0$ needs to satisfy
	\vspace{-5pt}
	\begin{equation}\label{eq_x0_ge}
		\small
		\max\left\{A-\frac{\lambda}{2\left|\tilde{\theta}_{21}^r\right|},0\right\}\le{x^0}\le{A}.
		\vspace{-2pt}
	\end{equation}
	
	For the case of $x<x^0$, the sufficient condition that equation \eqref{eq_deriv_le} is always greater than $0$ is given by
	\vspace{-1pt}
	\begin{equation}
		\small
		\tilde{\theta}_{21}^r\sin\left(\frac{2\pi}{\lambda}\tilde{\theta}_{21}^rx+\angle\mathcal{F}_{12}\right)\le0,
		\vspace{-3pt}
	\end{equation}
	i.e., 
	{\small
	\begin{align}\label{eq_deri_2}
		\!\left\{\begin{aligned}
			&\!\frac{\lambda}{\tilde{\theta}_{21}^r}\!\!\left(\!d_2\!-\!\frac{1}{2}\!-\!\frac{\angle\mathcal{F}_{12}}{2\pi}\!\right)\!\le{x}\le\!\frac{\lambda}{\tilde{\theta}_{21}^r}\!\!\left(\!d_2\!-\!\frac{\angle\mathcal{F}_{12}}{2\pi}\!\right)\!,\vartheta_1^r\!<\!\vartheta_2^r,\exists\hspace{1pt}d_2\!\in\!\mathbb{Z},\\
			&\!\frac{\lambda}{\tilde{\theta}_{21}^r}\!\!\left(\!d_2\!+\!\frac{1}{2}\!-\!\frac{\angle\mathcal{F}_{12}}{2\pi}\!\right)\!\le{x}\le\!\frac{\lambda}{\tilde{\theta}_{21}^r}\!\!\left(\!d_2\!-\!\frac{\angle\mathcal{F}_{12}}{2\pi}\!\right)\!,\vartheta_1^r\!>\!\vartheta_2^r,\exists\hspace{1pt}d_2\!\in\!\mathbb{Z}.
		\end{aligned}\right.
	\end{align}}%
	If there exists an integer $d_2$ such that the following inequalities hold, the achievable throughput $y_1\!\left(x\right)$ increases with $x$ within the receive region $\left[0,x^0\right)$:
	
	\vspace{-10pt}
	{\small
	\begin{align}
		\left\{\begin{aligned}
			\frac{x^0\tilde{\theta}_{21}^r}{\lambda}+\frac{\angle\mathcal{F}_{12}}{2\pi}\le{d_2}\le\frac{\angle\mathcal{F}_{12}}{2\pi}+\frac{1}{2},~\vartheta_1^r<\vartheta_2^r,\\
			\frac{\angle\mathcal{F}_{12}}{2\pi}-\frac{1}{2}\le{d_2}\le\frac{x^0\tilde{\theta}_{21}^r}{\lambda}+\frac{\angle\mathcal{F}_{12}}{2\pi},~\vartheta_1^r>\vartheta_2^r,
		\end{aligned}\right.
	\end{align}}%
	which also indicates that any movement of MA will deteriorate the system throughput. Similarly, $x^0$ needs to meet
	\vspace{-2pt}
	\begin{equation}\label{eq_x0_le}
		\small
		0\le{x^0}\le\min\left\{\frac{\lambda}{2\left|\tilde{\theta}_{21}^r\right|},A\right\}.
		\vspace{-1pt}
	\end{equation}
	
	In summary, if there exist integers $d_1$ and $d_2$ satisfying inequalities \eqref{eq_yx_dec} and \eqref{eq_yx_inc}, the optimal MA position is the initial antenna coordinate $x^0$, i.e., the antenna can not be moved. Furthermore, combining \eqref{eq_x0_ge} and \eqref{eq_x0_le} yields the following inequality that the initial antenna position should satisfy:\looseness=-1
	\vspace{-1pt}
	\begin{equation}
		\small
		\max\left\{A-\frac{\lambda}{2\left|\tilde{\theta}_{21}^r\right|},0\right\}\le{x^0}\le\min\left\{\frac{\lambda}{2\left|\tilde{\theta}_{21}^r\right|},A\right\}.
	\end{equation}
	The maximum size $A$ of moving region and the virtual AoAs $\vartheta_1^r$ and $\vartheta_2^r$ at the user shall meet
	\vspace{-1pt}
	\begin{equation}
		A\left|\tilde{\theta}_{21}^r\right|\le\lambda,
	\end{equation}
	which completes the proof.
\end{proof}

\subsubsection{Multiple-Path Case}
If $L\ge2$, the objective function of (P2.1) can be rewritten based on equation \eqref{eq_channel_2} as follows:
\begin{equation}\label{eq_C_multi}
	\small
	\tilde{C}_{L\ge2}=\tilde{T}\log_2\!\left(\!1\!+\!\frac{P_{\text{m}}}{\sigma^2}\sum_{n=1}^{N}\left|\sum_{l=1}^{L_k}\tau_l^*e^{j\frac{2\pi}{\lambda}\left(x\vartheta_l^r-\mathbf{t}_n^T\mathbf{p}_l\right)}\right|^2\right),
\end{equation}
where denote the term of channel power gain as $y_2\!\left(x\right)$, i.e.,
\begin{equation}
	\small
	y_2\!\left(x\right)\overset{\triangle}{=}\sum_{n=1}^{N}\left|\sum_{l=1}^{L_k}\tau_l^*e^{j\frac{2\pi}{\lambda}\left(x\vartheta_l^r-\mathbf{t}_n^T\mathbf{p}_l\right)}\right|^2\!\!=\sum_{n=1}^{N}\left|\sum_{l=1}^{L_k}\tilde{\tau}_{n,l}e^{j\frac{2\pi}{\lambda}x\vartheta_l^r}\right|^2,
\end{equation}
where $\tilde{\tau}_{n,l}\overset{\triangle}{=}\tau_l^*e^{-j\frac{2\pi}{\lambda}\mathbf{t}_n^T\mathbf{p}_l}$. Since $\tilde{T}$ is a positive concave function symmetric w.r.t. $x=x^0$, the region where the optimal MA position lies can be determined if the period of $y_2\!\left(x\right)$ is obtained. However, it is difficult to explicitly represent this period for $L\ge3$ as the AoAs $\theta_l^r$ and $\phi_l^r$ at the user are randomly distributed. Nevertheless, the approximation of this period can be derived by assuming quantized virtual AoAs. Specifically, let $X$ denote the period for $y_2\!\left(x\right)$, we have
\begin{align}\label{eq_y2_period}
	&y_2\!\left(x\right)\equiv{y_2}\!\left(x+X\right)\Leftrightarrow\nonumber\\
	&\sum_{n=1}^{N}\left|\sum_{l=1}^{L_k}\tilde{\tau}_{n,l}e^{j\frac{2\pi}{\lambda}x\vartheta_l^r}\right|^2{\equiv}~\sum_{n=1}^{N}\left|\sum_{l=1}^{L_k}\tilde{\tau}_{n,l}e^{j\frac{2\pi}{\lambda}\left(x+X\right)\vartheta_l^r}\right|^2.
\end{align}
Then, we consider the following sufficient conditions for \eqref{eq_y2_period} to hold:
\begin{align}\label{eq_suff_multi}
&\left|\sum_{l=1}^{L_k}\tilde{\tau}_{n,l}e^{j\frac{2\pi}{\lambda}x\vartheta_l^r}\right|^2{\equiv}~\left|\sum_{l=1}^{L_k}\tilde{\tau}_{n,l}e^{j\frac{2\pi}{\lambda}\left(x+X\right)\vartheta_l^r}\right|^2,\nonumber\\
\Leftrightarrow&\sum_{a=1}^{L}\sum_{b=1}^{L}\tilde{\tau}_{n,a}^*\tilde{\tau}_{n,b}e^{j\frac{2\pi}{\lambda}\tilde{\theta}_{ba}^rx}\equiv\sum_{a=1}^{L}\sum_{b=1}^{L}\tilde{\tau}_{n,a}^*\tilde{\tau}_{n,b}e^{j\frac{2\pi}{\lambda}\tilde{\theta}_{ba}^r\left(x+X\right)},\nonumber\\
\Leftrightarrow&\sum_{a=1}^{L}\sum_{b=1,b\ne{a}}^{L}\tilde{\tau}_{n,a}^*\tilde{\tau}_{n,b}\left(1-e^{j\frac{2\pi}{\lambda}\tilde{\theta}_{ab}^rX}\right)e^{j\frac{2\pi}{\lambda}\tilde{\theta}_{ab}^rx}\equiv0,\nonumber\\
\Leftrightarrow&~1-e^{j\frac{2\pi}{\lambda}\tilde{\theta}_{ab}^rX}\equiv0,~a,b\in\left\{1,\ldots,L\right\},\nonumber\\
\Leftrightarrow&~\frac{X}{\lambda}\tilde{\theta}_{ab}^r\in{\mathbb{Z}},~a,b\in\left\{1,\ldots,L\right\}.
\end{align}
Accordingly, it is noted that the period $X$ is the minimum positive number that ensures $\frac{X}{\lambda}\tilde{\theta}_{ab}^r$ to be an integer for $\forall a,b\in\left\{1,\ldots,L\right\}$. To facilitate the analysis on the region where the optimal solution to (P2.1) lies, we quantize the virtual AoAs with resolution $\kappa_0$, i.e., $\vartheta_l^r\in\left\{-1+\frac{2\kappa-1}{\kappa_0}\right\}_{1\le\kappa\le\kappa_0}$. Without loss of generality, it is further assumed that the virtual AoAs are sorted in a non-decreasing order, i.e., $\vartheta_1^r\le\vartheta_2^r\le\ldots\le\vartheta_L^r$. Denote $\vartheta_l^r=-1+\frac{2\kappa_l-1}{\kappa_0}$, which means that $\vartheta_l^r$ corresponds to the $\kappa_l$-th element in the quantized set of virtual AoAs. Thus, the difference of any adjacent virtual AoAs can be expressed as $\vartheta_{l+1}^r-\vartheta_l^r=\frac{2}{\kappa_0}\left(\kappa_{l+1}-\kappa_l\right)\overset{\triangle}{=}\frac{2\mu_l}{\kappa_0},~1\le{l}\le{L}-1$. To guarantee $\frac{X}{\lambda}\left(\vartheta_a^r-\vartheta_b^r\right),~\forall a,b\in\left\{1,\ldots,L\right\}$ being integers, by denoting $\mu^\star$ as the maximal common factor for $\left\{\mu_l\right\}_{1\le{l}\le{L}-1}$, the minimum period $X$ of $y_2\!\left(x\right)$ is given by
\begin{equation}
	X=\frac{\kappa_0\lambda}{2\mu^\star},
\end{equation}
which indicates that the maxima of $y_2\!\left(x\right)$ can be achieved within any region where $x$ spans at least one period $X$. In light of the above, the optimal solution $x^\star$ to (P2.1) must lie in the interval\\ $\left[\max\!\left\{0,\min\!\left\{x^0\!-\!\frac{X}{2},A\!-\!X\right\}\!\right\}\!,\min\!\left\{A,\max\!\left\{x^0\!+\!\frac{X}{2},X\right\}\!\right\}\right]$ under the assumption of quantized virtual AoAs, i.e.,

\vspace{-8pt}
{\small
\begin{align}\label{eq_inter_multipath}
	x^\star\in\bigg[&\max\left\{0,\min\left\{x^0-\frac{\kappa_0\lambda}{4\mu^\star},A-\frac{\kappa_0\lambda}{2\mu^\star}\right\}\right\},\nonumber\\
	&\quad\quad\quad~~\min\left\{A,\max\left\{x^0+\frac{\kappa_0\lambda}{4\mu^\star},\frac{\kappa_0\lambda}{2\mu^\star}\right\}\right\}\bigg].
\end{align}}%
Since $y_2\!\left(x\right)$ is expected to attain its maxima when $x$ is as close to $x^0$ as possible, which enables a larger optimal value for (P2.1), we thus present a more favorable design for the initial antenna coordinate $x^0$ as follows:
\begin{equation}\label{eq_ini_pos_multhpath}
	\frac{\kappa_0\lambda}{4\mu^\star}\le{x^0}\le{A}-\frac{\kappa_0\lambda}{4\mu^\star}.
\end{equation}
This design further indicates that the quantization resolution $\kappa_0$ and the length $A$ of moving region should satisfy
\begin{equation}\label{eq_nesu_multipath}
	\frac{\kappa_0}{A}\le\frac{2\mu^\star}{\lambda}.
\end{equation}

\vspace{-15pt}
\subsection{Where to Move?}\label{subsec_single_antenna}
While (P2.2) is much simplified compared to (P2), it is still non-convex due to that the objective function is not concave w.r.t. $x$. To address this challenge, we explore an efficient algorithm to optimize MA position via the SCA method. By introducing slack variables $\left\{q,u,w\right\}$, (P2.2) can be converted to
\begin{subequations}\label{eq_P2.3}
	\begin{alignat}{2}
		\text{\rm(P2.3)}:\quad&\underset{x,q,u,w}{\max}\quad&&T\log_2\left(1+\frac{u}{\sigma^2}\right)-\frac{qw}{v}\label{eq_P2.3_a}\\
		&~~\mathrm{s.t.}&&u\le{y_3}\!\left(x\right),\label{eq_P2.3_b}\\
		&&&w\ge\log_2\left(1+\frac{y_3\!\left(x\right)}{\sigma^2}\right),\label{eq_P2.3_c}\\
		&&&q\ge{x}-x^0,\label{eq_P2.3_d}\\
		&&&q\ge-\left(x-x^0\right),\label{eq_P2.3_e}\\
		&&&\eqref{eq_P2_c},\nonumber
	\end{alignat}
\end{subequations}
where we define
\begin{equation}\label{eq_y3}
	y_3\!\left(x\right)\overset{\triangle}{=}2P_{\text{m}}\sum_{a=1}^{L-1}\!\sum_{b=a+1}^{L}\!\!\left|\mathcal{F}_{ab}\right|\cos\!\left(\frac{2\pi}{\lambda}\tilde{\theta}_{ba}^rx+\angle\mathcal{F}_{ab}\right)+P_{\text{m}}\mathcal{G}.
\end{equation}
The equivalence between problems (P2.3) and (P2.2) can be verified by contradiction. However, the resulting term $qw$ in the objective function and $y_3\!\left(x\right)$ in the constraints \eqref{eq_P2.3_b} and \eqref{eq_P2.3_c} are neither convex nor concave and thus remain intractable. With given local points $\left\{q^i,w^i\right\}$ in the $i$-th iteration, a convex upper-bound surrogate function for the term $qw$ can be constructed as [\citen{SunY_surrogate_function}, eq. (101)]
\begin{equation}\label{eq_qw_ub}
	qw\le\frac{1}{2}\left(\frac{w^i}{q^i}q^2+\frac{q^i}{w^i}w^2\right).
\end{equation}
Thus, the objective function of (P2.3) becomes concave w.r.t $\left\{q,u,w\right\}$. To proceed, we tackle the constraint \eqref{eq_P2.3_b}, whereas the non-convexity of $y_3\!\left(x\right)$ prevents us from constructing its global lower bound through the first-order Taylor expansion. However, a concave lower-bound surrogate function for $y_3\!\left(x\right)$ can be obtained by applying the second-order Taylor expansion. Specifically, introducing a positive real number $\delta^{\text{lb}}$ such that $\delta^{\text{lb}}\ge\frac{\operatorname{d}^2y_3\left(x\right)}{\operatorname{d}x^2}$, the following inequality holds:
\vspace{-2pt}
\begin{equation}\label{eq_y3_lb}
	\small
	\!\!y_3\!\left(x\right)\ge{y_3}\!\left(x^i\right)\!+\!\frac{\operatorname{d}y_3\!\left(x^i\right)}{\operatorname{d}x}\!\left(x\!-\!x^i\right)\!-\!\frac{\delta^{\text{lb}}}{2}\!\left(x\!-\!x^i\right)^2\overset{\triangle}{=}y_3^{\text{lb},i}\!\left(x\right)\!,
\end{equation}
which is derived by modifying [\citen{SunY_surrogate_function}, Lemma 12]. In \eqref{eq_y3_lb}, $x^i$ denotes the given local point in the $i$-th iteration and $\frac{\operatorname{d}y_3\left(x^i\right)}{\operatorname{d}x}$ is given by
\vspace{-2pt}
\begin{equation}
	\small
	\!\!\!\frac{\operatorname{d}y_3\!\left(x^i\right)}{\operatorname{d}x}\!=\!-\frac{4{\pi}P_{\text{m}}}{\lambda}\!\sum_{a=1}^{L-1}\!\sum_{b=a+1}^{L}\!\!\left|\mathcal{F}_{ab}\right|\tilde{\theta}_{ba}^r\sin\!\left(\!\frac{2\pi}{\lambda}\tilde{\theta}_{ba}^rx^i\!\!+\!\angle\mathcal{F}_{ab}\!\right)\!,
	\vspace{-2pt}
\end{equation}
Then, a feasible $\delta^{\text{lb}}$ can be selected by calculating the following formulas

\begin{algorithm}[t]
	\renewcommand{\algorithmicrequire}{\textbf{Input:}}
	\renewcommand{\algorithmicensure}{\textbf{Output:}}
	\caption{Antenna Position Design}
	\begin{algorithmic}[1]
		\label{Alg_1}
		\REQUIRE Threshold $\epsilon>0$ and initial values $\left\{x^0,q^0,u^0,w^0\right\}$.
		\ENSURE The optimized solution $x^\star$.
		\STATE Set iteration number $i=1$.
		\REPEAT
		\STATE Update $\left\{x^{[i]},q^{[i]},u^{[i]},w^{[i]}\right\}$ by solving (P2.4).
		\STATE Update $i=i+1$.
		\UNTIL{The fractional increase in objective values between two consecutive iterations drops below the threshold $\epsilon$.}
	\end{algorithmic}
\end{algorithm}
\setlength{\textfloatsep}{10pt}

\vspace{-25pt}
{\small
\begin{align}\label{eq_delta}
	\frac{\operatorname{d}^2y_3\!\left(x\right)}{\operatorname{d}x^2}&=-\frac{8{\pi^2}P_{\text{m}}}{\lambda^2}\sum_{a=1}^{L-1}\!\sum_{b=a+1}^{L}\!\left|\mathcal{F}_{ab}\right|\left(\tilde{\theta}_{ba}^r\right)^2\!\!\cos\!\left(\frac{2\pi}{\lambda}\tilde{\theta}_{ba}^rx\!+\!\angle\mathcal{F}_{ab}\right)\nonumber\\
	&\le\frac{8{\pi^2}P_{\text{m}}}{\lambda^2}\sum_{a=1}^{L-1}\sum_{b=a+1}^{L}\left|\mathcal{F}_{ab}\right|\left(\tilde{\theta}_{ba}^r\right)^2\overset{\triangle}{=}\delta^{\text{lb}}.
\end{align}}%
With \eqref{eq_y3_lb}, a convex subset of constraint \eqref{eq_P2.3_b} is given by
\begin{equation}\label{eq_P2.3_b_SCA}
	u\le{y_3^{\text{lb},i}}\!\left(x\right).
\end{equation}
The sole remaining hurdle in addressing (P2.3) pertains to the non-convex constraint \eqref{eq_P2.3_c}. Note that the inequality \eqref{eq_y3_lb} is not applicable as a convex upper bound for $y_3\!\left(x\right)$ needs to be constructed. Nevertheless, according to [\citen{SunY_surrogate_function}, Lemma 12], the following inequality holds:
\vspace{-1pt}
\begin{equation}\label{eq_y3_ub}
	\small
	\!\!\!y_3\!\left(x\right)\le{y_3}\!\left(x^i\right)\!+\!\frac{\operatorname{d}y_3\!\left(x^i\right)}{\operatorname{d}x}\!\left(x\!-\!x^i\right)\!+\!\frac{\delta^\text{ub}}{2}\!\left(x\!-\!x^i\right)^2\!\overset{\triangle}{=}y_3^{\text{ub},i}\!\left(x\right)\!,
\end{equation}
where $\delta^{\text{ub}}$ is a positive real number satisfying $\delta^{\text{ub}}\ge\frac{\operatorname{d}^2y_3\left(x\right)}{\operatorname{d}x^2}$ and can be determined by following similar steps in \eqref{eq_delta}. Subsequently, since the left-hand-side (LHS) of \eqref{eq_P2.3_c} can be transformed into $\left(\sigma^22^w-\sigma^2\right)$, which is convex w.r.t $w$, it is lower-bounded by its first-order Taylor expansion, i.e.,
\vspace{-1pt}
\begin{equation}\label{eq_w_lb}
	\sigma^22^w-\sigma^2\ge\sigma^22^{w^i}-\sigma^2+\sigma^22^{w^i}\left(w-w^i\right)\ln2,
\end{equation}
where $w^i$ represents the given local point in the $i$-th iteration. With \eqref{eq_y3_ub} and \eqref{eq_w_lb}, a convex subset of constraint \eqref{eq_P2.3_c} is given by
\vspace{-3pt}
\begin{equation}\label{eq_P2.3_c_SCA}
	\sigma^22^{w^i}-\sigma^2+\sigma^22^{w^i}\left(w-w^i\right)\ln2\ge{y_3^{\text{ub},i}}\!\left(x\right).
	\vspace{-1pt}
\end{equation}
By replacing $qw$, constraints \eqref{eq_P2.3_b} and \eqref{eq_P2.3_c} in (P2.3) with \eqref{eq_qw_ub}, \eqref{eq_P2.3_b_SCA} and \eqref{eq_P2.3_c_SCA}, respectively, a lower bound for the optimal objective value of (P2.3) can be acquired by solving
\begin{subequations}\label{eq_P2.4}
	\begin{alignat}{2}
		\!\!\text{\rm(P2.4)}:~&\underset{x,q,u,w}{\max}~~&&T\log_2\left(1+\frac{u}{\sigma^2}\right)-\frac{\left(\frac{w^i}{q^i}q^2+\frac{q^i}{w^i}w^2\right)}{2v}\label{eq_P2.4_a}\\
		&~~\mathrm{s.t.}&&\eqref{eq_P2_c},\eqref{eq_P2.3_d},\eqref{eq_P2.3_e},\eqref{eq_P2.3_b_SCA},\eqref{eq_P2.3_c_SCA}.\nonumber
	\end{alignat}
\end{subequations}
According to the above derivations, the non-convex problem (P2.2) is ultimately recast as the convex antenna position optimization problem (P2.4), which can be solved optimally by off-the-shelf solvers (e.g., CVX \cite{Boyd_CVX_book}). The proposed algorithm for (P2.4) is outlined in \textbf{Algorithm \ref{Alg_1}}.

\vspace{5pt}
\section{Multiuser System}\label{section4}
In the following, we focus on the general multiuser scenario, where the duration $t_2$ of the IT phase depends on the maximal moving delay among all $K$ users. To this end, we propose an efficient algorithm to cope with (P1) in an alternating manner. The problem of interest can be reformulated as
\begin{subequations}\label{eq_P3}
	\begin{alignat}{2}
		\text{(P3)}:\quad&\underset{\eta,\mathbf{x},\left\{\mathbf{w}_k\right\}_{k\in\mathcal{K}}}{\max}\quad&&\eta\label{eq_P3_a}\\
		&~~\quad\mathrm{s.t.}&&\tilde{C}_k\ge\eta,~{\forall}k\in\mathcal{K},\label{eq_P3_b}\\
		&&&\eqref{eq_P1_b},\eqref{eq_P1_c},\nonumber
	\end{alignat}
\end{subequations}
where $\eta$ is the introduced auxiliary variable representing the minimum achievable throughput among all users. Note that $\mathbf{x}$ and $\left\{\mathbf{w}_k\right\}_{k\in\mathcal{K}}$ are intricately coupled in the LHS of constraint \eqref{eq_P3_b}, which renders \eqref{eq_P3_b} to be non-convex and the resultant problem (P3) to be non-convex. This presents challenges in seeking the optimal solution. Then, by dividing (P3) into two subproblems, an alternating optimization (AO) algorithm is proposed, where the transmit beamforming and MA positions are updated alternately.

\subsection{Transmit Beamforming Design}\label{subsec_trans}
First, we consider the subproblem of transmit beamforming optimization with given MA positions. Introduce slack variables $\left\{\alpha_k,\beta_k\right\}_{k\in\mathcal{K}}$ such that
\begin{align}
	&e^{\alpha_k}=\frac{\sum_{j=1}^{K}\left|\mathbf{h}_k^H\mathbf{w}_j\right|^2}{\sigma_k^2}+1,~\forall{k}\in\mathcal{K},\label{eq_alpha}\\
	&e^{\beta_k}=\frac{\sum_{j=1,j{\ne}k}^{K}\left|\mathbf{h}_k^H\mathbf{w}_j\right|^2}{\sigma_k^2}+1,~\forall{k}\in\mathcal{K}.\label{eq_beta}
\end{align}
The LHS of \eqref{eq_P3_b} can be rewritten as $\tilde{T}_\text{m}\log_2e^{\left(\alpha_k-\beta_k\right)}$, where we define $\tilde{T}_\text{m}\overset{\triangle}{=}T-\underset{k\in\mathcal{K}}{\max}\left\{\frac{\left|x_k-x_k^0\right|}{v_k}\right\}$. Consequently, for any given MA positions $\mathbf{x}$, the subproblem of (P3) regarding to $\left\{\eta,\left\{\mathbf{w}_k,\alpha_k,\beta_k\right\}_{k\in\mathcal{K}}\right\}$ can be formulated as
\begin{subequations}\label{eq_P3.1}
	\begin{alignat}{2}
		\text{(P3.1)}:&\underset{\eta,\left\{\mathbf{w}_k,\alpha_k,\beta_k\right\}_{k\in\mathcal{K}}}{\max}&&~\eta\label{eq_P3_1_a}\\
		&\mathrm{s.t.}&&\hspace{-33pt}\tilde{T}_\text{m}\log_2e^{\left(\alpha_k-\beta_k\right)}\ge\eta,~{\forall}k\in\mathcal{K},\label{eq_P3_1_b}\\
		&&&\hspace{-35pt}\sum_{j=1}^{K}\left|\mathbf{h}_k^H\mathbf{w}_j\right|^2+\sigma_k^2\ge\sigma_k^2e^{\alpha_k},~{\forall}k\in\mathcal{K},\label{eq_P3_1_c}\\
		&&&\hspace{-43pt}\sum_{j=1,j{\ne}k}^{K}\!\!\!\left|\mathbf{h}_k^H\mathbf{w}_j\right|^2+\sigma_k^2\le\sigma_k^2e^{\beta_k},~{\forall}k\in\mathcal{K},\label{eq_P3_1_d}\\
		&&&\hspace{-34pt}\eqref{eq_P1_b}.\nonumber
	\end{alignat}
\end{subequations}
Note that constraints \eqref{eq_P3_1_c} and \eqref{eq_P3_1_d} are obtained by replacing the equality signs in \eqref{eq_alpha} and \eqref{eq_beta} with inequality signs. At the optimal solution to problem (P3.1), the constraints \eqref{eq_P3_1_c} and \eqref{eq_P3_1_d} must be active, since otherwise \eqref{eq_P3_1_b} can take the inequality sign by increasing $\left\{\alpha_k\right\}_{k\in\mathcal{K}}$ or decreasing $\left\{\beta_k\right\}_{k\in\mathcal{K}}$ such that the objective value can be further improved. However, the convexity of the LHS of \eqref{eq_P3_1_c} and the right-hand-side (RHS) of \eqref{eq_P3_1_d} results in the non-convexity of (P3.1). This prompts us to tackle the constraints w.r.t. $\left\{\mathbf{w}_k\right\}_{k\in\mathcal{K}}$ by applying the semidefinite relaxation (SDR) technique \cite{WuQQ_IRS_active_passive_trans,GaoY_activeIRS_SWIPT}. Specifically, denote $\mathbf{H}_k\overset{\triangle}{=}\mathbf{h}_k\mathbf{h}_k^H$ and define $\mathbf{W}_k\overset{\triangle}{=}\mathbf{w}_k\mathbf{w}_k^H$ for $k\in\mathcal{K}$, which needs to satisfy $\mathbf{W}_k\succeq\mathbf{0}$ and $\operatorname{rank}\left(\mathbf{W}_k\right)=1$. By utilizing $\mathbf{w}_k^H\mathbf{h}_k\mathbf{h}_k^H\mathbf{w}_k=\operatorname{tr}\left(\mathbf{H}_k\mathbf{W}_k\right)$ and relaxing the rank-one constraint on $\left\{\mathbf{W}_k\right\}_{k\in\mathcal{K}}$, (P3.1) is recast as
\begin{subequations}\label{eq_P3.2}
	\begin{alignat}{2}
		\text{(P3.2)}:&\underset{\eta,\left\{\mathbf{W}_k\in\mathbb{H}^N,\alpha_k,\beta_k\right\}_{k\in\mathcal{K}}}{\max}&&\eta\label{eq_P3_2_a}\\
		&\mathrm{s.t.}&&\hspace{-62pt}\sum_{j=1}^{K}\operatorname{tr}\left(\mathbf{H}_k\mathbf{W}_j\right)+\sigma_k^2\ge\sigma_k^2e^{\alpha_k},~{\forall}k\in\mathcal{K},\label{eq_P3_2_b}\\
		&&&\hspace{-70pt}\sum_{j=1,j{\ne}k}^{K}\!\!\!\!\operatorname{tr}\left(\mathbf{H}_k\mathbf{W}_j\right)+\sigma_k^2\le\sigma_k^2e^{\beta_k},~{\forall}k\in\mathcal{K},\label{eq_P3_2_c}\\
		&&&\hspace{-62pt}\sum_{k=1}^{K}\operatorname{tr}\left(\mathbf{W}_k\right)\le{P_{\text{m}}},\label{eq_P3_2_d}\\
		&&&\hspace{-61pt}\mathbf{W}_k\succeq\mathbf{0},~{\forall}k\in\mathcal{K},\label{eq_P3_2_e}\\
		&&&\hspace{-61pt}\eqref{eq_P3_1_b}.\nonumber
	\end{alignat}
\end{subequations}

Then, the remaining non-convex constraint \eqref{eq_P3_2_c} of problem (P3.2) can be approximated by its first-order Taylor expansion-based affine under-estimator, yielding
\begin{subequations}\label{eq_P3.3}
	\begin{alignat}{2}
		\text{(P3.3)}:&\underset{\eta,\left\{\mathbf{W}_k\in\mathbb{H}^N,\alpha_k,\beta_k\right\}_{k\in\mathcal{K}}}{\max}&&~\eta\label{eq_P3_3_a}\\
		&\hspace{-22pt}\mathrm{s.t.}&&\hspace{-93pt}\sum_{j=1,j{\ne}k}^{K}\!\!\!\!\operatorname{tr}\left(\mathbf{H}_k\mathbf{W}_j\right)+\sigma_k^2\le\sigma_k^2~\!y_4^{\text{lb},i}\!\left(\beta_k\right),~\!{\forall}k\in\mathcal{K},\label{eq_P3_3_b}\\
		&&&\hspace{-84pt}\eqref{eq_P3_1_b},\eqref{eq_P3_2_b},\eqref{eq_P3_2_d},\eqref{eq_P3_2_e},\nonumber
	\end{alignat}
\end{subequations}
where
\begin{equation}
	y_4^{\text{lb},i}\!\left(\beta_k\right)\overset{\triangle}{=}\beta_ke^{\beta_k^i}+\left(1-\beta_k^i\right)e^{\beta_k^i}
\end{equation}
with $\left\{\beta_k^i\right\}_{k\in\mathcal{K}}$ being the given local points in the $i$-th iteration. As problem (P3.3) is a convex semidefinite program (SDP), its optimal solution can be efficiently obtained by standard solvers.

\begin{remark}\label{remark}
	In the subsequent subsection, we will utilize the previously obtained $\left\{\mathbf{W}_k\right\}_{k\in\mathcal{K}}$ and incorporate them into all computations. Only after the final iteration will we recover the transmit beamforming vectors $\left\{\mathbf{w}_k\right\}_{k\in\mathcal{K}}$ based on $\left\{\mathbf{W}_k\right\}_{k\in\mathcal{K}}$, which ensures the convergence of the proposed algorithm. The discussion of recovering $\left\{\mathbf{w}_k\right\}_{k\in\mathcal{K}}$ from $\left\{\mathbf{W}_k\right\}_{k\in\mathcal{K}}$ will be presented at the end of the next subsection.
\end{remark}

\vspace{-10pt}
\subsection{Antenna Position Design}
With fixed Hermitian semidefinite matrices $\left\{\mathbf{W}_k\right\}_{k\in\mathcal{K}}$, and recalling the previous definitions $e^{\alpha_k}$, $e^{\beta_k}$, and $y_4^{\text{lb},i}\!\left(\beta_k\right)$ in Subsection \ref{subsec_trans}, the subproblem of (P3) regarding to $\mathbf{x}$ can be rewritten in a similar way as
\begin{subequations}\label{eq_P3.4}
	\begin{alignat}{2}
		\text{(P3.4)}:&~\underset{\mathbf{x},\eta,\zeta,\left\{\alpha_k,\beta_k\right\}_{k\in\mathcal{K}}}{\max}~&&\eta\label{eq_P3_4_a}\\
		&\hspace{-19pt}\mathrm{s.t.}&&\hspace{-69pt}\left(T-\zeta\right)\log_2e^{\left(\alpha_k-\beta_k\right)}\ge\eta,~{\forall}k\in\mathcal{K},\label{eq_P3_4_b}\\
		&&&\hspace{-66pt}\zeta\ge\frac{x_k-x_k^0}{v_k},~{\forall}k\in\mathcal{K},\label{eq_P3_4_c}\\
		&&&\hspace{-66pt}\zeta\ge-\frac{x_k-x_k^0}{v_k},~{\forall}k\in\mathcal{K},\label{eq_P3_4_d}\\
		&&&\hspace{-67pt}\sum_{j=1}^{K}\mathbf{h}_k^H\mathbf{W}_j\mathbf{h}_k+\sigma_k^2\ge\sigma_k^2e^{\alpha_k},~{\forall}k\in\mathcal{K},\label{eq_P3_4_e}\\
		&&&\hspace{-75pt}\sum_{j=1,j{\ne}k}^{K}\!\!\!\!\mathbf{h}_k^H\mathbf{W}_j\mathbf{h}_k+\sigma_k^2\le\sigma_k^2~\!y_4^{\text{lb},i}\!\left(\beta_k\right),~{\forall}k\in\mathcal{K},\label{eq_P3_4_f}\\
		&&&\hspace{-66pt}\eqref{eq_P1_c}.\nonumber
	\end{alignat}
\end{subequations}
Note that problem (P3.4) is non-convex and challenging to solve as variables $\left\{\zeta,\left\{\alpha_k,\beta_k\right\}_{k\in\mathcal{K}}\right\}$ are coupled in constraint \eqref{eq_P3_4_b} and $\mathbf{x}$ is not explicit in \eqref{eq_P3_4_e} and \eqref{eq_P3_4_f}, which motivates the development of the following algorithm. First, the LHS of \eqref{eq_P3_4_b} is not jointly concave w.r.t. $\zeta$ and $\left\{\alpha_k,\beta_k\right\}_{k\in\mathcal{K}}$ but satisfies [\citen{SunY_surrogate_function}, eq. (101) and (102)]
\begin{align}
	\zeta\alpha_k&\le\frac{1}{2}\left(\frac{\alpha_k^i}{\zeta^i}\zeta^2+\frac{\zeta^i}{\alpha_k^i}\alpha_k^2\right)\overset{\triangle}{=}y_5^{\text{ub},i}\!\left(\zeta,\alpha_k\right),\label{eq_y5}\\
	\zeta\beta_k&\ge\left(1+\ln\zeta+\ln\beta_k-\ln\zeta^i-\ln\beta_k^i\right)\zeta^i\beta_k^i\nonumber\\&\overset{\triangle}{=}y_6^{\text{lb},i}\!\left(\zeta,\beta_k\right).\label{eq_y6}
\end{align}
where $\zeta^i$ and $\left\{\alpha_k^i\right\}_{k\in\mathcal{K}}$ are the given local points in the $i$-th iteration. Then, a convex approximation of constraint \eqref{eq_P3_4_b} can be established immediately, given by
\begin{equation}\label{eq_P3_4_b_re}
	y_6^{\text{lb},i}\!\left(\zeta,\beta_k\right)\ge{y}_5^{\text{ub},i}\!\left(\zeta,\alpha_k\right)+\frac{\eta}{\log_2e}+\beta_kT-\alpha_kT.
\end{equation}

By expanding the term $\mathbf{h}_k^H\mathbf{W}_j\mathbf{h}_k$ in the LHS of both \eqref{eq_P3_4_e} and \eqref{eq_P3_4_f}, we derive the following reconstructed equation

\vspace{-10pt}
{\small
\begin{align}
	&y_{kj}\!\left(x_k\right)\overset{\triangle}{=}\mathbf{h}_k^H\mathbf{W}_j\mathbf{h}_k=\operatorname{tr}\left(\mathbf{A}_{kj}\right)+\nonumber\\
	&\quad~2\!\sum_{a=1}^{L_k-1}\!\!\sum_{b=a+1}^{L_k}\!\!\left|\mathbf{A}_{kj}\!\left(a,b\right)\right|\cos\!\left(\tilde{\vartheta}_{k,ba}^rx_k\!+\!\angle\mathbf{A}_{kj}\!\left(a,b\right)\right),
\end{align}}%
where we define
\begin{align}
	&\mathbf{A}_{kj}\overset{\triangle}{=}\mathbf{\Delta}_k^H\mathbf{G}_k\mathbf{W}_j\mathbf{G}_k^H\mathbf{\Delta}_k,\nonumber\\
	&\tilde{\vartheta}_{k,ba}^r\overset{\triangle}{=}\frac{2\pi}{\lambda}\left(\vartheta_{k,b}^r-\vartheta_{k,a}^r\right),~a,b\in\left\{1,\ldots,L\right\}.\nonumber
\end{align}
Following the same approach as in \ref{subsec_single_antenna}, the SCA method is employed to address the non-convex constraints \eqref{eq_P3_4_e} and \eqref{eq_P3_4_f} based on the formulas above. By applying the second-order Taylor expansion, we construct quadratic surrogate functions that serve as a global lower bound for $y_{kj}\!\left(x_k\right)$ in \eqref{eq_P3_4_e} and upper bound for $y_{kj}\!\left(x_k\right)$ in \eqref{eq_P3_4_f}. Specifically, given the local points $\left\{x_k^i\right\}_{k\in\mathcal{K}}$ obtained in the $i$-th iteration, the following inequality holds \cite{WangHH_interference_MA}, [\citen{SunY_surrogate_function}, Lemma 12]:
%
\begin{align}
	y_{kj}\!\left(x_k\right)&\ge{y}_{kj}^{\text{lb},i}\!\left(x_k\right)\nonumber\\
	&\hspace{-11pt}\overset{\triangle}{=}y_{kj}\!\left(x_k^i\right)\!+\!\frac{\operatorname{d}y_{kj}\!\left(x_k^i\right)}{\operatorname{d}x_k}\!\left(\!x_k\!-\!x_k^i\!\right)\!-\!\frac{\delta_{kj}}{2}\!\left(\!x_k\!-\!x_k^i\!\right)^2,\label{eq_y_kj_SCA_lb}\\
	y_{kj}\!\left(x_k\right)&\le{y}_{kj}^{\text{ub},i}\!\left(x_k\right)\nonumber\\
	&\hspace{-11pt}\overset{\triangle}{=}y_{kj}\!\left(x_k^i\right)\!+\!\frac{\operatorname{d}y_{kj}\!\left(x_k^i\right)}{\operatorname{d}x_k}\!\left(\!x_k\!-\!x_k^i\!\right)\!+\!\frac{\delta_{kj}}{2}\!\left(\!x_k\!-\!x_k^i\!\right)^2,\label{eq_y_kj_SCA_ub}
\end{align}
where $\delta_{kj}$ is the introduced positive real number such that $\delta_{kj}\ge\frac{\operatorname{d}^2y_{kj}\left(x_k\right)}{\operatorname{d}x^2}$, and $\frac{\operatorname{d}y_{kj}\left(x_k^i\right)}{\operatorname{d}x_k}$ is given by
\begin{equation}
	\footnotesize
	\frac{\operatorname{d}y_{kj}\!\left(x_k^i\right)}{\operatorname{d}x_k}=-2\!\sum_{a=1}^{L_k\!-\!1}\!\!\sum_{b=a\!+\!1}^{L_k}\!\!\tilde{\vartheta}_{k,ba}^r\!\left|\mathbf{A}_{kj}\!\left(a,b\right)\right|\sin\!\left(\!\tilde{\vartheta}_{k,ba}^rx_k^i\!+\!\angle\mathbf{A}_{kj}\!\left(a,b\right)\!\right).
\end{equation}
Then, a feasible auxiliary value $\delta_{kj}$ can be derived as

\vspace{-10pt}
{\footnotesize
\begin{align}
	\frac{\operatorname{d}^2\!y_{kj}\!\left(x_k\right)}{\operatorname{d}x^2}&=-2\!\sum_{a=1}^{L_k\!-\!1}\!\!\sum_{b=a\!+\!1}^{L_k}\!\!\tilde{\vartheta}{_{k,ba}^r}\!^2\!\left|\mathbf{A}_{kj}\!\left(a,b\right)\right|\!\cos\!\left(\tilde{\vartheta}_{k,ba}^rx_k\!+\!\angle\mathbf{A}_{kj}\!\left(a,b\right)\!\right)\nonumber\\
	&\le2\sum_{a=1}^{L_k-1}\!\!\sum_{b=a+1}^{L_k}\!\!\tilde{\vartheta}{_{k,ba}^r}\!^2\left|\mathbf{A}_{kj}\!\left(a,b\right)\right|\overset{\triangle}{=}\delta_{kj},~\forall{k,j}\in\mathcal{K}.
\end{align}}%
In this way, by tightening the LHS of \eqref{eq_P3_4_e} to its lower bound and that of \eqref{eq_P3_4_f} to upper bound, (P3.4) is then evidently reduced to the following convex optimization problem (P3.5) for the given local points $\left\{\mathbf{x}^i,\zeta^i\left\{\alpha_k^i,\beta_k^i\right\}_{k\in\mathcal{K}}\right\}$ obtained in the $i$-th iteration:
\begin{subequations}\label{eq_P3.5}
	\begin{alignat}{2}
		\text{(P3.5)}:&~\underset{\mathbf{x},\eta,\zeta,\left\{\alpha_k,\beta_k\right\}_{k\in\mathcal{K}}}{\max}~&&\eta\label{eq_P3_5_a}\\
		&\hspace{-21pt}\mathrm{s.t.}&&\hspace{-70pt}\sum_{j=1}^{K}y_{kj}^{\text{lb},i}\!\left(x_k\right)+\sigma_k^2\ge\sigma_k^2e^{\alpha_k},~\forall{k}\in\mathcal{K},\label{eq_P3_5_b}\\
		&&&\hspace{-78pt}\sum_{j=1,j\ne{k}}^{K}y_{kj}^{\text{ub},i}\!\left(x_k\right)+\sigma_k^2\le\sigma_k^2~\!y_4^{\text{lb},i}\!\left(\beta_k\right),~{\forall}k\in\mathcal{K},\label{eq_P3_5_e}\\
		&&&\hspace{-69pt}\eqref{eq_P1_c},\eqref{eq_P3_4_c},\eqref{eq_P3_4_d},\eqref{eq_P3_4_b_re}.\nonumber
	\end{alignat}
\end{subequations}
Similarly, convex optimization solvers, such as CVX \cite{Boyd_CVX_book}, can be utilized to solve (P3.5) optimally.


After the final iteration of the proposed algorithm, reminiscing about \textbf{Remark \ref{remark}}, the relaxed problem (P3.3) may not lead to a rank-one solution, i.e., $\operatorname{rank}\left(\mathbf{W}_k\right)\ne1,~\forall{k}\in\mathcal{K}$. This implies that the optimal objective value of (P3.3) only serves as an upper bound of (P3.1). Thus, additional steps are required to construct a rank-one solution $\left\{\mathbf{w}_k\right\}_{k\in\mathcal{K}}$ from the obtained higher-rank solution $\left\{\mathbf{W}_k\right\}_{k\in\mathcal{K}}$, the details of which can be found in \cite{WuQQ_IRS_active_passive_conf} and are therefore omitted here. By substituting the transmit beamforming vectors $\left\{\mathbf{w}_k\right\}_{k\in\mathcal{K}}$ sourced from the previous construction and the MA positions $\mathbf{x}$ obtained by solving (P3.5) in the last iteration back into (P3), we ultimately attain the optimized minimum achievable throughput $\eta$.

\subsection{Convergence and Computational Complexity Analysis}
\begin{algorithm}[t]
	\renewcommand{\algorithmicrequire}{\textbf{Input:}}
	\renewcommand{\algorithmicensure}{\textbf{Output:}}
	\caption{Joint Antenna Position and Beamforming Design}
	\begin{algorithmic}[1]
		\label{Alg_2}
		\REQUIRE Threshold $\epsilon>0$, initial values $\left\{\mathbf{x}^0,\zeta^0,\left\{\beta_k^0\right\}_{k\in\mathcal{K}}\right\}$.
		\ENSURE The optimized solution $\left\{\eta^\star,\mathbf{x}^\star,\left\{\mathbf{w}_k^\star\right\}_{k\in\mathcal{K}}\right\}$.
		\STATE Set iteration number $i=1$.
		\REPEAT
		\STATE Update $\left\{\mathbf{W}_k^i,\hat{\alpha}_k^i,\hat{\beta}_k^i\right\}_{k\in\mathcal{K}}$ by solving (P3.3) with $\left\{\mathbf{x}^{i-1},\left\{\beta_k^{i-1}\right\}_{k\in\mathcal{K}}\right\}$.
		\STATE Update $\left\{\mathbf{x}^i,\eta^i,\zeta^i,\left\{\alpha_k^i,\beta_k^i\right\}_{k\in\mathcal{K}}\right\}$ by solving (P3.5) with $\left\{\mathbf{x}^{i-1},\zeta^{i-1},\left\{\mathbf{W}_k^i,\hat{\alpha}_k^i,\hat{\beta}_k^i\right\}_{k\in\mathcal{K}}\right\}$.
		\STATE Update $i=i+1$.
		\UNTIL The fractional increase in objective values of problem (P3.5) between two consecutive iterations falls below the threshold $\epsilon$.
		\STATE Construct $\left\{\mathbf{w}_k^\star\right\}_{k\in\mathcal{K}}$ from $\left\{\mathbf{W}_k^{i-1}\right\}_{k\in\mathcal{K}}$ and obtain $\eta^\star$ by solving (P3) with fixed $\left\{\mathbf{x}^\star=\mathbf{x}^{i-1},\left\{\mathbf{w}_k^\star\right\}_{k\in\mathcal{K}}\right\}$.
	\end{algorithmic}
\end{algorithm}
\setlength{\textfloatsep}{10pt}
Based on the above derivations, the proposed algorithm for the multiuser system is summarized in \textbf{Algorithm \ref{Alg_2}}. Next, a proof of convergence is provided as follows. We define $\eta_1\left(\mathbf{x},\left\{\mathbf{W}_k,\alpha_k,\beta_k\right\}_{k\in\mathcal{K}}\right)$ and $\eta_2\left(\mathbf{x},\zeta,\left\{\mathbf{W}_k,\alpha_k,\beta_k\right\}_{k\in\mathcal{K}}\right)$ as the objective values of problems (P3.3) and (P3.5), respectively. Then, in the iterative steps $3-5$ of \textbf{Algorithm \ref{Alg_2}}, we have\looseness=-1
\begin{align}
	&\eta_1\left(\mathbf{x}^{i-1},\left\{\mathbf{W}_k^i,\hat{\alpha}_k^i,\hat{\beta}_k^i\right\}_{k\in\mathcal{K}}\right)\nonumber\\
	\overset{(a)}{=}~&\eta_2\left(\mathbf{x}^{i-1},\zeta^{i-1},\left\{\mathbf{W}_k^i,\hat{\alpha}_k^i,\hat{\beta}_k^i\right\}_{k\in\mathcal{K}}\right)\nonumber\\
	\overset{(b)}{\le}~&\eta_2\left(\mathbf{x}^i,\zeta^i,\left\{\mathbf{W}_k^i,\alpha_k^i,\beta_k^i\right\}_{k\in\mathcal{K}}\right)\nonumber\\
	\overset{(c)}{=}~&\eta_1\left(\mathbf{x}^i,\left\{\mathbf{W}_k^i,\alpha_k^i,\beta_k^i\right\}_{k\in\mathcal{K}}\right)\nonumber\\
	\overset{(d)}{\le}~&\eta_1\left(\mathbf{x}^i,\left\{\mathbf{W}_k^{i+1},\hat{\alpha}_k^{i+1},\hat{\beta}_k^{i+1}\right\}_{k\in\mathcal{K}}\right),
\end{align}
where $(a)$ holds since one of the constraints \eqref{eq_P3_4_c} or \eqref{eq_P3_4_d} in (P3.5) is active at given $\left\{\mathbf{x}^{i-1},\zeta^{i-1}\right\}$, i.e., $\zeta^{i-1}=\underset{k\in\mathcal{K}}{\max}\left\{\frac{\left|x_k^{i-1}-x_k^0\right|}{v_k}\right\}$, and \eqref{eq_y5}, \eqref{eq_y6} as well as the second-order Taylor expansions \eqref{eq_y_kj_SCA_lb} and \eqref{eq_y_kj_SCA_ub} are all tight at given $\left\{\mathbf{x}^{i-1},\zeta^{i-1},\left\{\mathbf{W}_k^i,\hat{\alpha}_k^i,\hat{\beta}_k^i\right\}_{k\in\mathcal{K}}\right\}$; $(b)$ and $(d)$ hold due to that $\left\{\mathbf{x}^i,\zeta^i,\left\{\mathbf{W}_k^i,\alpha_k^i,\beta_k^i\right\}_{k\in\mathcal{K}}\right\}$ and $\left\{\mathbf{x}^i,\left\{\mathbf{W}_k^{i+1},\hat{\alpha}_k^{i+1},\hat{\beta}_k^{i+1}\right\}_{k\in\mathcal{K}}\right\}$ are the optimal solutions to (P3.5) and (P3.3), respectively; $(c)$ holds since (P3.3) is feasible and has the same objective value as (P3.5) at given $\left\{\mathbf{x}^i,\left\{\mathbf{W}_k^i,\alpha_k^i,\beta_k^i\right\}_{k\in\mathcal{K}}\right\}$. Thus, we obtain a non-decreasing sequence of the objective values for (P3.5) by repeating steps $3-5$ in \textbf{Algorithm \ref{Alg_2}}. Additionally, it can be easily proved that the optimal objective value of (P3.5) is upper-bounded. Consequently, the convergence of \textbf{Algorithm \ref{Alg_2}} is assured.

The computational complexity of \textbf{Algorithm \ref{Alg_2}} mainly lies in solving the SDP problem (P3.3) and (P3.5). According to the complexity analyses in \cite{Polik_interior_point_book} and \cite{WangKY_complexity_analysis}, the computational costs for solving (P3.3) and (P3.5) are $\mathcal{O}\left(\ln\frac{1}{\epsilon}\sqrt{KN}K^{3}N^{6}\right)$ and $\mathcal{O}\left(\ln\frac{1}{\epsilon}\sqrt{K}K^{3}\right)$, respectively. Hereto, the computational complexity of \textbf{Algorithm \ref{Alg_2}} is approximately $\mathcal{O}\left(I\ln\frac{1}{\epsilon}K^{3.5}N^{6.5}\right)$, where $I$ and $\epsilon$ represent the number of iterations required for convergence and the prescribed solution accuracy, respectively.

\section{Numerical Results}\label{section5}
This section provides numerical examples to validate the effectiveness of our proposed algorithms. We consider a three-dimensional (3D) coordinate system where the BS equipped with $4\times4~\left(N=16\right)$ uniform planar array (UPA) is located at $\left(0,0,0\right)$ and the users equipped with linear MA arrays are randomly dispersed on a hemispherical plane with radius of $r_k=r=100~\text{meters~(m)},~\forall{k}\in\mathcal{K}$ from the BS. The channel model in \eqref{eq_channel_1} is adopted here, where the number of channel paths for each user is assumed to be identical, i.e., $L_k=L,~\forall{k}\in\mathcal{K}$. Under this condition, it is considered that the diagonal elements of PRM $\mathbf{\Delta}_k$ in \eqref{eq_channel_1} follow the CSCG distribution $\mathcal{CN}\!\left(0,\Gamma_k^2/L\right)$, where $\Gamma_k^2=\rho_0r^{-\xi_0}$ is the expected channel power gain of $\mathbf{h}_{k}$, $\rho_0=-42~\text{dB}$ represents the expected value of the average channel power gain at the reference distance of $1~\text{m}$, and $\xi_0=2.8$ denotes the path-loss exponent. To ensure the fairness of comparison, the total power of the diagonal elements in PRM is set to be the same for the channels with different numbers of paths, i.e., $\mathbb{E}\left\{\operatorname{tr}\left(\mathbf{\Delta}_{k}^H\mathbf{\Delta}_{k}\right)\right\}\equiv{\Gamma_k^2},~\forall{L}$. The elevation and azimuth AoDs/AoAs of the channel paths for each user are assumed to be independent and identically distributed variables, following a uniform distribution over $\left[0,\pi\right]$. The moving speeds and regions of receive MAs at users are identically set as $v_k=v$ and $\mathcal{A}_k=\mathcal{A}=\left[0,A\right],~\forall{k}\in\mathcal{K}$, respectively. Besides, the initial antenna positions are similarly set as $x_k^0=x^0=A/2$. Other adopted settings of simulation parameters are $\sigma_k^2=\sigma^2=-80~\text{dBm}$ and $\epsilon=10^{-4}$. The values of $K,~L,~A,~v,~T,~P_\text{m},~\kappa_0$ will be specified in the following simulations.

\subsection{Single-User System}
First, we consider a single-user scenario, i.e., $K=1$, with antenna moving speed $v=v_\text{1}=0.1~\text{m/s}$ and $v=v_\text{2}=0.2~\text{m/s}$. By varying the system setup, i.e., the number of channel paths $L$, the length of moving region $A$, and the duration of transmission block $T$, we investigate the throughput achieved within a restricted duration of transmission block, under the given maximum transmit power of BS $P_\text{m}=10~\text{dBm}$. The results obtained by our proposed algorithm in Section \ref{subsec_single_antenna} are termed as \textbf{Algorithm \ref{Alg_1}} and other benchmark schemes for comparison are defined as follows: 1) \textbf{Quantized}: according to Section \ref{subsec_analysis}, virtual AoAs at the user are quantized with a quantization resolution $\kappa_0=10$, based on which \textit{Algorithm \ref{Alg_1}} is utilized to optimize the MA position. 2) \textbf{Max SNR}: MA of the user is deployed at the position that maximizes channel power gain and the detailed optimization algorithm can be referred to \cite{GaoY_multicast_MA}. 3) \textbf{FPA}: the receive antenna of the user is fixed at the initial position and the entire duration of a transmission block is allocated to the IT phase.

In Fig. \ref{Fig_2}, we evaluate the convergence performance of \textit{Algorithm \ref{Alg_1}} and compare it with the method in \cite{GaoY_multicast_MA} of maximizing the user's SNR. The number of channel paths, the moving region size, and the transmission block duration are set as $L=10$, $A=2\lambda$, and $T=3~\text{s}$, respectively. As can be observed, the achievable throughput and channel power gain respectively attained by \textit{Algorithm \ref{Alg_1}} and method in \cite{GaoY_multicast_MA} both exhibit an upward trend with the increasing iteration index. The curves demonstrate particularly rapid growth in the early stage, stabilizing after approximately $30$ (\textit{Algorithm \ref{Alg_1}}) and $50$ (\textit{Max SNR} scheme) iterations. By applying our proposed algorithm, the achievable throughput of the user rises from $5.25~\text{bits/Hz}$ to $5.75~\text{bits/Hz}$, which suggests its superior performance over \textit{Max SNR} scheme, as will be specifically illustrated next.\looseness=-1

\begin{figure}[t]
	\centering
	\includegraphics[width=3.2in]{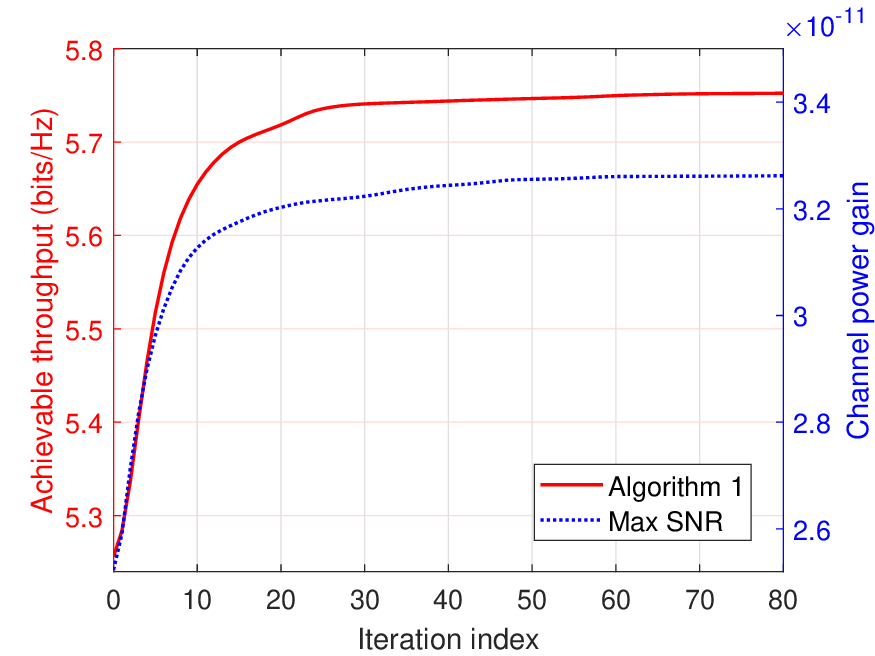}
	\caption{Convergence behaviors of \textit{Algorithms \ref{Alg_1}} and method in \cite{GaoY_multicast_MA}.}
	\label{Fig_2}
	\vspace{4pt}
	\centering
	\includegraphics[width=3.2in]{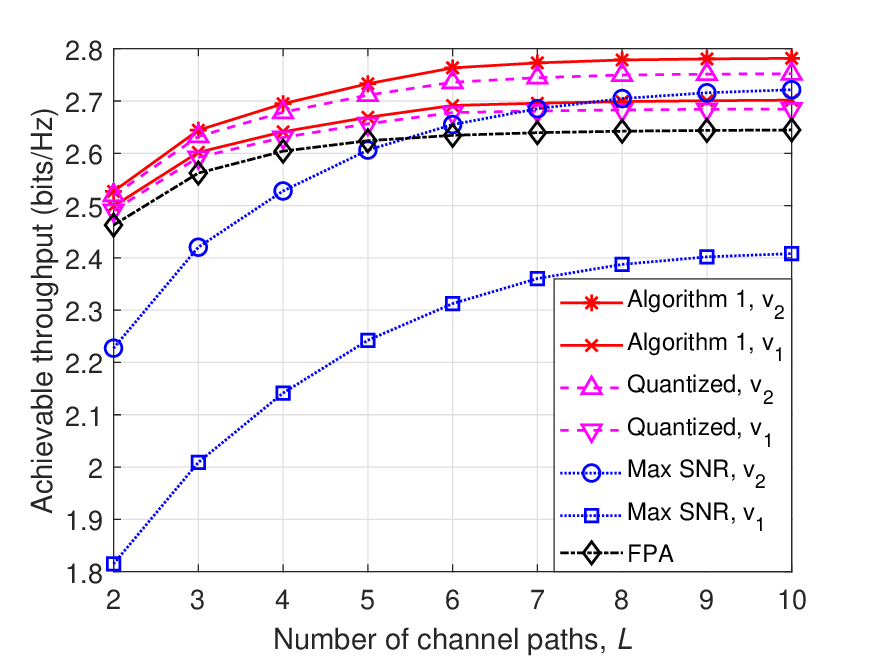}
	\caption{Achievable throughput versus number of channel paths.}
	\label{Fig_3}
	\vspace{4pt}
\end{figure}

In Fig. \ref{Fig_3}, we elaborate on the relationship between the achievable throughput and the number of channel paths and conduct comparative analyses of \textit{Algorithm \ref{Alg_1}} and the benchmark schemes. The parameters are set as $A=2\lambda$ and $T=1.5~\text{s}$. It is observed that with the increase in $L$, all the schemes, including \textit{FPA} scheme, perform an enhancement in achievable throughput. This improvement is ascribed to the heightened multi-path diversity gain with increased $L$. Note that \textit{Max SNR} scheme does not even achieve the throughput of \textit{FPA} scheme when $L\le5$, which is due to the optimized MA position in \textit{Max SNR} scheme only guarantees a high channel power gain but leads to a farther distance away from the initial position (compressing IT phase $t_2$ and therefore reducing the throughput). Additionally, the variation of antenna moving speed has a greater impact on \textit{Max SNR} scheme than \textit{Algorithm \ref{Alg_1}} and \textit{Quantized} scheme. This is attributed to the fact that faster speed allows AM delay $t_1$ to be shorter, mitigating the adverse effect caused by the squeezed IT duration, which reveals that \textit{Max SNR} scheme is speed-sensitive (this will be fully demonstrated in Fig. \ref{Fig_6}). As $L$ becomes larger, the small-scale fading undergoes more pronounced fluctuations in the receive region, making it easier for \textit{Max SNR} scheme to reach a local optimum of channel power gain within a short distance of antenna movement. Hence, \textit{Max SNR} scheme possesses a more arresting performance growth with $L$. Thanks to the thorough consideration of the trade-off between MA position corresponding to disparate channel conditions and AM delay, \textit{Quantized} scheme can approximate the throughput of perfect angle information while quantizing the virtual AoAs with a resolution of only $\kappa_0=10$. This indicates that \textit{Algorithm \ref{Alg_1}} does not require high accuracy of angle estimation, making it possible to deploy it in real systems.

\begin{figure}[t]
	\centering
	\includegraphics[width=3.2in]{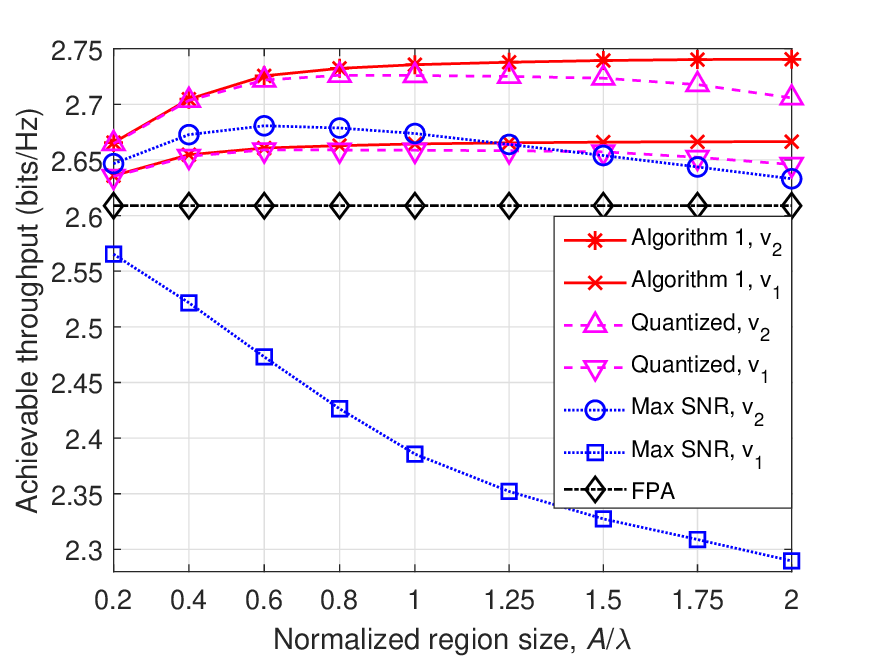}
	\caption{Achievable throughput versus normalized region size.}
	\label{Fig_4}
	\vspace{4pt}
	\centering
	\includegraphics[width=3.2in]{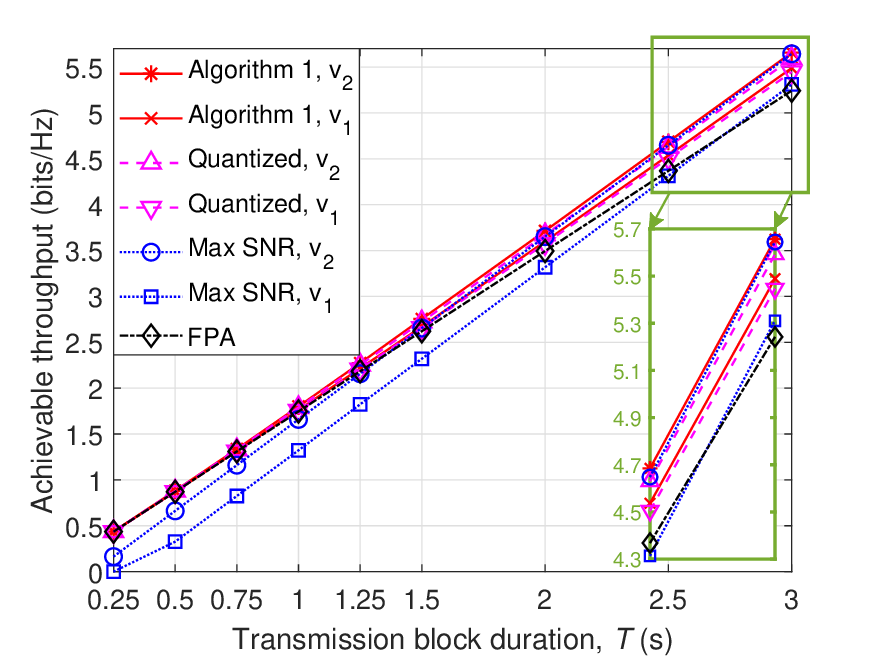}
	\caption{Achievable throughput versus transmission block duration.}
	\label{Fig_5}
	\vspace{4pt}
\end{figure}

Fig. \ref{Fig_4} depicts the achievable throughput obtained by different schemes versus the normalized region size for MA moving when $L=6$ and $T=1.5~\text{s}$, where the size of moving region is normalized by carrier wavelength, i.e., $A/\lambda$. Unlike the results of most MA works, the throughput of \textit{Max SNR} scheme that optimizes MA position only for channel power gain within a restricted transmission block duration does not grow all the time with moving region size. This is due to the larger the region for MA moving, the more likely it is that the optimized MA position in \textit{Max SNR} scheme falls at a coordinate farther away from the initial position. Thus, a longer distance of antenna movement will severely squeeze the IT duration and result in decreased throughput. This phenomenon becomes more significant as the antenna moving speed is low. In contrast, although \textit{Quantized} scheme also shows performance degradation when $A/\lambda\ge1.75$, it still approximates \textit{Algorithm \ref{Alg_1}} well, especially when $A/\lambda$ is small. From the observation, the throughput achieved by \textit{Algorithm \ref{Alg_1}}, \textit{Quantized}, and \textit{Max SNR} schemes at high speed $v_\text{2}$ with $A/\lambda\le0.6$ exhibits a monotonic growth, which results from the enhanced flexibility in optimizing MA position within enlarged region, leading to a potential performance improvement. However, beyond a certain threshold of $A/\lambda$, further increase of it does not bring any performance gains to \textit{Algorithm \ref{Alg_1}}, which suggests that the optimized achievable throughput is attainable within a relatively small receive region. Moreover, the throughput of \textit{Max SNR} scheme at low speed $v_\text{1}$ is consistently inferior to that of \textit{FPA} scheme, and the performance of \textit{Max SNR} scheme at high speed is gradually worse than that of \textit{Quantized} scheme at low speed. This indicates that \textit{Max SNR} scheme is not stable for different moving region sizes and is even not available if the antenna moving speed is low, whereas \textit{Algorithm \ref{Alg_1}} and \textit{Quantized} scheme are.

Fig. \ref{Fig_5} investigates the impact of transmission block duration on the achievable throughput of various schemes when $L=6$ and $A=2\lambda$. It is observed that the performance of \textit{Max SNR} scheme is extremely poor compared to other schemes with short duration $T$, and the throughput of that even drops to $0$ at low speed $v_\text{1}$. This is ascribed to the fact that the ultra-short duration of the transmission block cannot afford any AM delay, therefore no occurrence of antenna movement, i.e., \textit{FPA} scheme, is the optimal design in this case. Note that the curves of \textit{Max SNR} scheme gradually approach that of \textit{Algorithm \ref{Alg_1}}, which is owing to the proportion of AM delay within the entire transmission block reduces with increased $T$, leading to the convergence of solutions from \textit{Max SNR} scheme and \textit{Algorithm \ref{Alg_1}}. Moreover, since \textit{Algorithm \ref{Alg_1}} and \textit{Quantized} scheme are highly adaptive for different $T$, the throughput of which converges to that of \textit{FPA} scheme for small $T$ and gains performance boosts of up to approximately $8\%$ higher than \textit{FPA} scheme as $T$ increases (the boosts are expected to rise further when $T>3~\text{s}$). This fact corroborates the high robustness of \textit{Algorithm \ref{Alg_1}} to low accuracy of angle estimation. To sum up, the optimal schemes for ultra-short and ultra-long durations are respectively \textit{FPA} scheme and \textit{Max SNR} scheme, while \textit{Algorithm \ref{Alg_1}} incorporates the advantages of both two, maintaining excellent performance for arbitrary $T$.

\subsection{Multiuser System}
Next, we consider a multiuser system, i.e., $K>1$, with $L=6$, $A=2\lambda$, and $T=1.5~\text{s}$ and compare the minimum achievable throughput of our proposed algorithm in Section \ref{section4} that is termed as \textbf{Algorithm \ref{Alg_2}} with four benchmark schemes defined as follows: 1) \textbf{Quantized 1}: the virtual AoAs are quantized with a quantization resolution $\kappa_0=10$, based on which \textit{Algorithm \ref{Alg_2}} is utilized to optimize the MA positions. 2) \textbf{Quantized 2}: the quantization resolution is set as $\kappa_0=20$ and \textit{Algorithm \ref{Alg_2}} is employed similarly. 3) \textbf{Max min SINR}: the MAs are deployed at the positions that maximize their minimum SINR and the detailed optimization algorithm can be referred to \cite{GaoY_multicast_MA}. 4) \textbf{FPA}: the transmit beamforming is obtained by Section \ref{subsec_trans} and other settings are similar to the previous subsection. Besides, the number of random vectors used for Gaussian randomization is set to be $1000$.

\begin{figure}[t]
	\centering
	\includegraphics[width=3.2in]{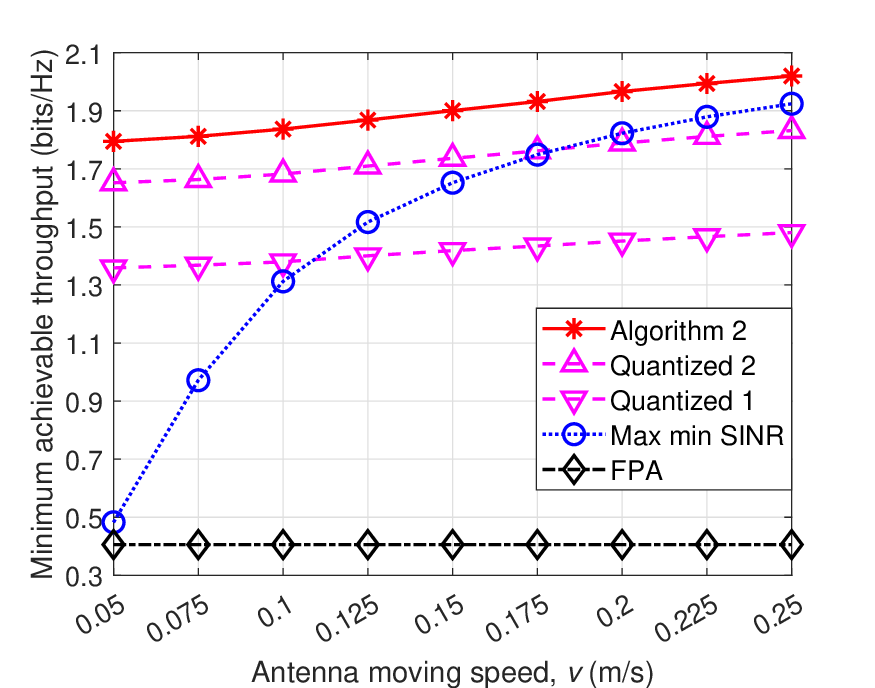}
	\caption{Minimum achievable throughput versus antenna moving speed.}
	\label{Fig_6}
	\vspace{7pt}
	\centering
	\includegraphics[width=3.2in]{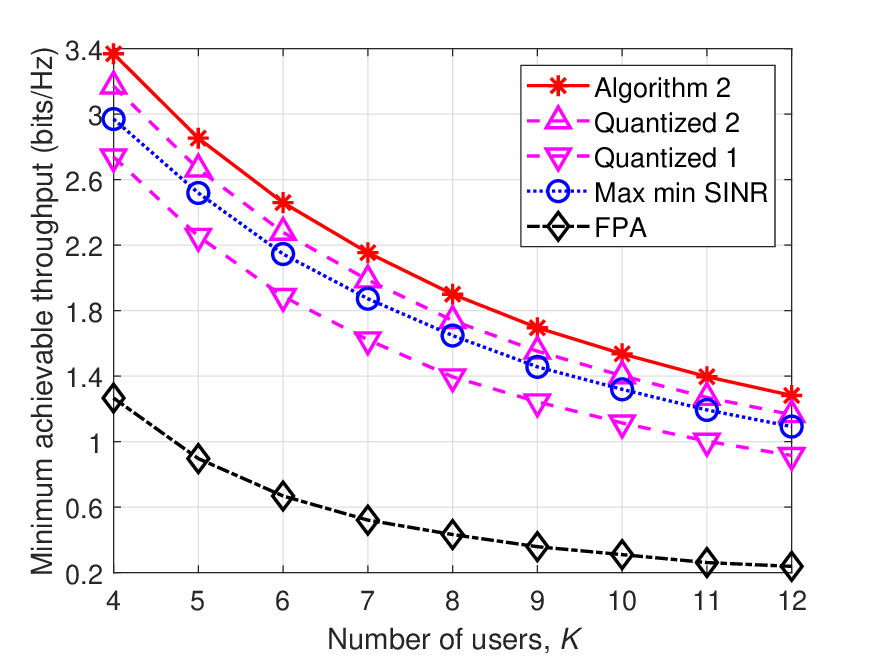}
	\caption{Minimum achievable throughput versus number of users.}
	\label{Fig_7}
	\vspace{4pt}
\end{figure}

In Fig. \ref{Fig_6}, we illustrate the minimum achievable throughput of various schemes versus the antenna moving speed, with parameters set to $K=8$ and $P_{\text{m}}=20~\text{dBm}$. It is evident that \textit{Algorithm \ref{Alg_2}}, \textit{Quantized 1}, and \textit{Quantized 2} schemes exhibit robust adaptability to varying speeds $v$, with their performance remaining relatively stable regardless of any speed. Furthermore, it is noteworthy that the quantization resolution $\kappa_0=10$ of virtual AoAs, as employed in the previous subsection, demonstrates relatively poor performance in multiuser scenarios. However, by merely increasing the resolution from $10$ to $20$, approximately $90\%$ of the performance achieved with perfect angle information can be attained. Consequently, \textit{Algorithm \ref{Alg_2}} exhibits commendable stability and adaptability w.r.t. antenna moving speed and angle estimation accuracy. In contrast, \textit{Max min SINR} scheme approximates the performance of \textit{Algorithm \ref{Alg_2}} only at high speeds, while at speeds below $0.1~\text{m/s}$, it is inferior to \textit{Quantized 1} scheme and even approaches the performance of \textit{FPA} scheme. This indicates that \textit{Max min SINR} scheme is highly sensitive to antenna moving speed, imposing stringent requirements on the movement units of MAs. Therefore, prior to the introduction of our proposed algorithm, existing designs were unable to accommodate MA-enabled communication systems with a limited transmission block duration.

Fig. \ref{Fig_7} plots the minimum achievable throughput obtained by different schemes versus the number of users when $v=0.15~\text{m/s}$ and $P_{\text{m}}=20~\text{dBm}$. It is observed that there is a significant decrease in the minimum achievable throughput with increasing $K$ for all schemes. This phenomenon is attributed to the heightened inter-user interference and diminished per-user transmit power. Nevertheless, irrespective of the scheme deployed, the communication performance of MA-aided systems markedly surpasses that of \textit{FPA} scheme, underscoring the efficacy of MA in significantly mitigating interference in multiuser systems. Furthermore, as the user count escalates, the performance enhancement of \textit{Algorithm \ref{Alg_2}} over \textit{FPA} scheme amplifies from an initial $166\%$ to $438\%$. Even for \textit{Quantized 1} scheme, with a virtual AoAs quantization resolution of merely $\kappa_0=10$, the performance improvement over \textit{FPA} scheme reaches $284\%$. By augmenting the resolution $\kappa_0$ from $10$ to $20$, our proposed algorithm demonstrates superior performance under quantized virtual AoAs compared to \textit{Max min SINR} scheme that necessitates precise angle information. This augmentation culminates in a $389\%$ increase in minimum achievable throughput relative to \textit{FPA} scheme.

\begin{figure}[t]
	\centering
	\includegraphics[width=3.2in]{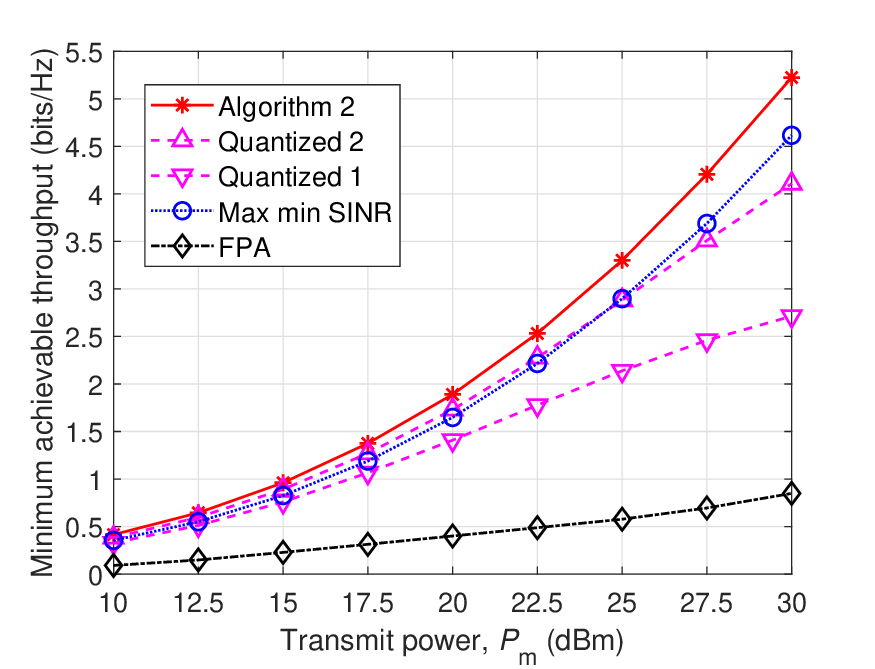}
	\caption{Minimum achievable throughput versus transmit power.}
	\label{Fig_8}
	\vspace{4pt}
\end{figure}

The effect of transmit power $P_\text{m}$ at BS on the minimum achievable throughput of different schemes is shown in Fig. \ref{Fig_8}. Other parameters are set as $K=8$ and $v=0.15~\text{m/s}$. With increased $P_\text{m}$, the performance of MA schemes exhibits a higher growth rate compared to \textit{FPA} scheme. This is expected as MAs reduce inter-user interference, allowing the increase in transmit power to significantly improve SINR at users. In contrast, increasing $P_\text{m}$ in FPA systems amplifies unexpected interference, resulting in less noticeable SINR improvement. Notably, at low transmit power, our proposed algorithm with a lower quantization resolution of virtual AoAs ($\kappa_0=10$) can achieve performance comparable to \textit{Algorithm \ref{Alg_2}} and \textit{Max min SINR} scheme, which utilize precise angle information. This indicates that \textit{Quantized 1} scheme, which requires only low accuracy in angle estimation, can be considered for deployment at low transmit power. As $P_\text{m}$ increases, the superior interference suppression enabled by accurate angle information becomes increasingly pronounced, leading to a more significant performance enhancement in \textit{Algorithm \ref{Alg_2}} and \textit{Max min SINR} scheme compared to others. Particularly, \textit{Algorithm \ref{Alg_2}} strikes an excellent balance between the channel conditions and AM delays, making its performance outstanding among all schemes.

\section{Conclusion}\label{section6}
In this paper, we modeled the minimum achievable throughput within a transmission block of limited duration, wherein the antenna moving delay was considered for the first time in the research field of MA. To maximize this throughput in MA-enabled multiuser communications, we jointly optimized MA positions and transmit beamforming. Despite the non-convex nature of the problem, characterized by highly coupled optimization variables, we analyzed the impact of antenna movement on communication performance and proposed a high-efficiency SCA-based algorithm. Initially tailored for the single-user scenario, this algorithm was then generalized to the multiuser context. Simulation results substantiated the effectiveness of the proposed algorithms, offering valuable engineering insights.

\bibliographystyle{IEEEtran}
\bibliography{bibfile}
\end{document}